\documentclass{article}

\usepackage{bbold}
\usepackage{lscape}
\usepackage{amsmath}
\usepackage{amsthm}
\usepackage{amssymb,stmaryrd}
\usepackage{xspace}
\usepackage{mathrsfs}

\usepackage{pgf, pgfcore, tikz} 
\usetikzlibrary{arrows} 
\usetikzlibrary{automata,positioning}
\usetikzlibrary{decorations.pathmorphing}
\usetikzlibrary{shapes.geometric}

\usepackage{subfig}

\usepackage{listings}
\usepackage{algpseudocode}
\usepackage{algorithm}

\newtheorem{theorem}{Theorem}
\newtheorem{definition}[theorem]{Definition}

\newtheorem{proposition}[theorem]{Proposition}
\newtheorem{lemma}[theorem]{Lemma}
\newtheorem{corollary}[theorem]{Corollary}

\newtheoremstyle
{exemple}
{\topsep}
{\topsep}
{\fontseries{lt}\selectfont}
{}
{\bfseries \scshape}
{.}
{ }
{\thmname{#1}\thmnumber{ #2}\thmnote{ #3}}

\theoremstyle{exemple}
\newtheorem{exemple}[theorem]{Example}

\newcommand{\hmid}{\ \mid\ }
\newcommand{\XML}{XML\xspace}
\newcommand{\Figure}{Figure\xspace} 
\newcommand{\ol}[1]{\overline{#1}} 

\newcommand{\AutA}{\mathcal{A}} 
\newcommand{\AutB}{\mathcal{B}} 
\newcommand{\La}[1]{L(#1)} 

\newcommand{\AutAGeneral}{\AutA = (Q, \Sigma, Q_f, \Delta)} 

\newcommand{\Lh}[1]{{\cal H}_{#1}} 
\newcommand{\LhA}{\Lh{\AutA}}
\newcommand{\SLh}{S_{\AutA}} 
\newcommand{\BLh}{\AutB_{\AutA}} 

\newcommand{\racineRun}[1]{\underset{#1}{\hookrightarrow}} 
\newcommand{\pathLabeled}[1]{\overset{#1}{\leadsto}}
\newcommand{\run}{r}

\newcommand{\subsetatc}[1]{\left\lfloor #1 \right\rfloor}

\newcommand{\T}[1]{T_{#1}}
\newcommand{\TSigma}{\T{\Sigma}}
\newcommand{\height}[1]{{\it height}(#1)}
\newcommand{\Pos}[1]{{\it nodes}(#1)}
\newcommand{\Node}{p} 
\newcommand{\Root}{\epsilon} 
\newcommand{\subtree}[2]{{#1}_{|#2}}
\newcommand{\lin}[1]{[#1]}

\newcommand{\Hdg}[1]{H_{#1}}
\newcommand{\HdgSigma}{\Hdg{\Sigma}}
\newcommand{\open}[1]{{\sf open}(#1)}  

\newcommand{\PostOp}{\it Post}
\newcommand{\Post}[2]{\PostOp_{#1}(#2)}
\newcommand{\Posta}[1]{\Post{a}{#1}}
\newcommand{\PostAllStar}[1]{\PostOp^*(#1)}
\newcommand{\MinPostAllStar}[1]{\PostOp^*_{\lfloor \rfloor}(#1)}
\newcommand{\PostAll}[1]{\PostOp(#1)}
\newcommand{\MinPostAll}[1]{\PostOp_{\lfloor \rfloor}(#1)}
\newcommand{\PostAlli}[2]{\PostOp^{#1}(#2)}
\newcommand{\MinPostAlli}[2]{\PostOp^{#1}_{\lfloor \rfloor}(#2)}
\newcommand{\PointFixeEtoile}{\Pi_{\HdgSigma}}   
\newcommand{\PointFixe}{P_{\TSigma}}   

\newcommand{\Pt}[1]{P_{#1}} 
\newcommand{\pih}[1]{\pi_{#1}} 

\newcommand{\rel}[1]{{\sf rel}(#1)}
\newcommand{\relOp}{\sf rel}
\newcommand{\id}{{\sf id}} 

\newcommand{\macroSet}{\mathscr{P}} 
\newcommand{\relSet}{\mathscr{R}} 
\newcommand{\mswSet}{\mathscr{W}} 

\newcommand{\wordsOf}[1]{\overline{#1}}

\newcommand{\X}{X} 
\newcommand{\Y}{Y} 
\newcommand{\Pref}[1]{{\it Pref}(#1)} 

\newcommand{\PPref}[1]{{\it PPref}(#1)} 

\newcommand{\VPA}{VPA\xspace}
\newcommand{\VPAs}{VPAs\xspace}
\newcommand{\vparule}[4]{
  {\ensuremath{#1}}
  \xrightarrow{\ensuremath{#2}:\ensuremath{#3}}
  {\ensuremath{#4}}
}
\newcommand{\stack}{\sigma} 
\newcommand{\ConfSet}{{\mathscr C}}
\newcommand{\Safe}[1]{{\it Safe}(#1)}
\newcommand{\LSafe}[1]{{\it LSafe}(#1)}
\newcommand{\Reach}[1]{{\it Reach}(#1)}
\newcommand{\Pred}[2]{{\it Pred}_{#1}(#2)}
\newcommand{\finiteH}{H} 
\newcommand{\vparel}[1]{{\sf rel}_{#1}} 
\newcommand{\vpaid}{{\it id}} 
\newcommand{\vpaheq}{\sim} 
\newcommand{\roundcup}{\sqcup} 
\newcommand{\satform}[2]{\varphi_{#1}(#2)} 
\newcommand{\charact}[1]{v_{#1}} 
\newcommand{\charactval}{v} 

\title{Visibly pushdown automata on trees:\\ universality and $u$-universality}
\author{V\'eronique Bruy\`ere \and Marc Ducobu \and Olivier Gauwin}
\date{}

\begin{document}
\maketitle

\begin{abstract}
An automaton is universal if it accepts every possible input.
We study the notion of $u$-universality,
which asserts that the automaton accepts every input starting with $u$.
Universality and $u$-universality are both EXPTIME-hard
for non-deterministic tree automata.
We propose efficient antichain-based techniques to address these problems for 
visibly pushdown automata operating on trees.
One of our approaches yields algorithms for the universality
and $u$-universality of hedge automata.
\end{abstract}

\section{Introduction}

The model-checking framework provided many successful tools
for decades, starting from the seminal work of B\"uchi. 
A lot of them rely on the links between logics used to express
properties on words, and automata allowing to check them.
Some of these results have been adapted to trees,
and more recently to words with a nesting structure.

\emph{Visibly pushdown automata} (\VPAs) have been introduced to process such 
words with nesting \cite{AlurMadhusudan04}. 
\VPAs are similar to pushdown automata,
but operate on a partitioned alphabet: 
a given letter is associated with one action (push or pop), and thus
cannot push when firing a transition, 
and pop when firing another.
Such automata were introduced to express and check properties
on control flows of programs, 
where procedure calls push on the stack, and returns pop. 
They are also suitable to express properties on \XML documents 
\cite{KumarMadhusudanViswanathan07}.
These are usually represented as trees, but are serialized as a sequence
of opening and closing tags, also called the \emph{linearization}
of this document, or its corresponding \emph{\XML stream}.

Processing such streams without building the corresponding tree
is permitted by online algorithms.
It is often crucial to detect \emph{at the earliest position} of the stream
whether it satisfies a given property or not.
When the property is given by an automaton,
we call this automaton \emph{$u$-universal}
when the stream begins with word $u$,
and $u$ ensures that the whole stream is accepted by the automaton,
whatever it contains after $u$.
Indeed, this is a variant of universality of automata:
universality is $\epsilon$-universality,
and amounts to assert that the property will be true for every possible stream,
and thus can be asserted before reading the first letter.
While universality of automata is a very strong property, 
$u$-universality arises each time
an automaton checks the \emph{presence of a pattern} in trees,
and this pattern appears in $u$.

A delay in the detection of a violation may be exploited to
break firewalling systems
when they use \XML for logs \cite{BenediktJeffreyLey-WildJL08},
or to perform a denial of service attack on a remote program.
In a less critical sense, it can also be used in \XML validators, 
to assert validation or non-validation of a document before reading it entirely.
For program traces, this is usually addressed by online verification
algorithms operating on words but without considering the nesting relation
between program calls and returns \cite{KupfermanVardi01}.
In the \XML setting, some streaming algorithms have been proposed.
Most of them are not earliest, and require a delay between the position
where acceptance/refusal can be decided, and the position where it is claimed.

Indeed, testing $u$-universality is computationally hard on 
linearizations of trees.
When the property is specified by a deterministic automaton,
this can be checked in cubic time.
On non-deterministic automata,
$u$-universality becomes EXPTIME-complete \cite{GauwinNiehrenTison09b}.
Non-determinism naturally arises when automata are obtained
from logic formulas, as for instance XPath expressions
with descendant axis \cite{FrancisDavidLibkin11,GauwinNiehren11}.

In this paper we propose new algorithms for deciding universality
and $u$-universality of non-deterministic tree automata
on unranked trees accessed through their linearization.
Our goal is to obtain algorithms that outperform
the usual approach consisting in determinizing the automaton.
We want our algorithms for $u$-universality to be incremental,
in that, for a letter $a$, deciding the $ua$-universality
should reuse as much information from $u$-universality computation as possible.
Indeed we want to find the earliest position allowing to assert acceptance,
so we have to test $u$-universality for every prefix $u$ before that point.

We use \emph{antichains} to get smaller objects to manipulate,
and develop other ad-hoc methods.
Antichains have been applied recently to decision problems related to 
non-deterministic automata:
universality and inclusion for finite word automata 
\cite{DewulfDoyenHenzinger06},
and for non-deterministic bottom-up tree automata
\cite{BouajjaniHHTV08}.
Some simulation relations are also known on unranked trees \cite{Srba06}
but it is unclear whether they can help for our problems,
as they do in other contexts \cite{AbdullaChenHolik10,DoyenRaskin10}.
Nguyen \cite{VanNguyen09} proposed an algorithm for testing the universality
of \VPAs. This algorithm simultaneously performs an on-the-fly determinization and reachability checking by $\cal P$-automaton. The notion of $\cal P$-automaton introduced in \cite{EsparzaHRS00,EsparzaKS03} provides a symbolic technique to compute the sets of all reachable configurations of a \VPA.
This algorithm has been later improved by Nguyen and Ohsaki
\cite{VanNguyenOhsaki12} by introducing antichains of over transitions of $\cal P$-automaton, in a way to generate reachable configurations as small as possible.
Our algorithms for universality are alternative to this one since we do not use the regularity property of the set  of reachable configurations. And our techniques for incrementally testing $u$-universality
are totally new wrt this algorithm.
A problem similar to $u$-universality is addressed in 
\cite{MadhusudanViswanathan09} in the context of query answering.
Their algorithm applies to non-deterministic \VPAs 
recognizing a canonical language of a query, 
but the automata are assumed to only accept prefixes $u$
for which $u$-universality holds, which is precisely the goal of our algorithms.

We contribute two algorithms for checking $u$-universality 
of \VPAs on linearizations of unranked trees.
The first algorithm is by reduction to $u$-universality (and also universality)
of hedge automata.
\emph{Hedge automata} are the standard automaton model
used for unranked trees \cite{BruggemannkleinMurataWood01},
and runs in a bottom-up manner.
Hedge automata are similar to \XML schema models like DTDs or Relax NG.
The second algorithm is a direct algorithm on \VPAs.
Such an algorithm was known in the deterministic case 
\cite{GauwinNiehrenTison09b}, and relied on the incremental
computation of safe states.
This algorithm cannot be generalized
to the non-deterministic case, as sets of safe states do not contain 
enough information.
Instead, we use sets of safe configurations, 
which may be infinite, but manipulated through finite antichains.
We show how SAT solvers can be used to update these antichains.

The paper is structured as follows.
In Section~\ref{sec:preliminaries} we define trees,
visibly pushdown automata and the problem of $u$-universality.
Section~\ref{sec:hedge-approach} details our first algorithm,
relying on a translation to hedge automata.
Section~\ref{sec:safe} contains our second algorithm,
namely the incremental computation of sets of safe configurations.

\section{Trees, Automata and $u$-universality}
\label{sec:preliminaries}

\subsection{Unranked Trees}
We recall here the standard definition of unranked trees, as 
provided for instance in \cite{tata}.
Let $\Sigma$ be a finite \emph{alphabet},
and $\Sigma^*$ (resp. $\Sigma^+$) be the set of all words 
(resp. non empty words) over $\Sigma$. The empty word is denoted by $\epsilon$.
Given two words $v, w \in \Sigma^*$  over $\Sigma$,
$v$ is a \emph{prefix} (resp. \emph{proper prefix}) of $w$
if there exists a word $v' \in \Sigma^*$ (resp. $v' \in \Sigma^+$) such that
$v v' = w$.
Let $\mathbb{N}_0$ be the set of all non-negative integers.

An \emph{unranked tree} $t$ over $\Sigma$ is a partial function 
$t: \mathbb{N}_0^* \rightarrow \Sigma$
such that the domain is non-empty, finite and prefix-closed.
The domain is denoted by $\Pos{t}$
and contains the \emph{nodes} of the tree $t$, 
with the root being the empty word $\Root$. 
The function $t$ labels each node $\Node$ with a letter $t(\Node)$ of $\Sigma$.
A node labeled by $a \in \Sigma$ is called an \emph{$a$-node}.
The set of all unranked trees over $\Sigma$ is denoted by
$\T{\Sigma}$.

The \emph{subtree} of $t$ rooted at node $\Node$ of $t$ is the tree denoted by 
$\subtree{t}{\Node}$, which domain is the set of nodes $\Node'$ such that
$\Node \Node'\in\Pos{t}$ and verifying 
$\subtree{t}{\Node}(\Node') = t(\Node \Node')$.
For a given node $\Node\in\Pos{t}$, we call \emph{children} of $\Node$
the nodes $\Node i \in\Pos{t}$ for $i\in\mathbb{N}_0$,
and use the usual definitions for parents, ancestors and descendants.
The \emph{height} of a tree is the length of its longest branch 
(with the length being the number of nodes).

\begin{exemple}
\label{def_ex_arbres}
Let $t_1: \{ \Root, 1, 2, 3, 4, 5, 51, 52  \} 
\rightarrow \{a, b, c\}$ such that 
$t_1(\Root) = c$, 
$t_1(1) = a$, $t_1(2) = a$, $t_1(3) = a$,
$t_1(4) = a$, $t_1(5) = b$, $t_1(51) = b$, $t_1(52) = b$.
Tree $t_1$ is an unranked tree with height $3$. It can be represented as in 
\Figure~\ref{exemple_arbre_1}.

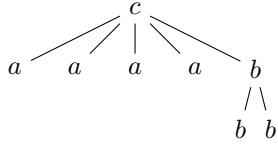
\begin{figure}[h!]
\centering
\begin{tikzpicture}
[level distance=8mm, 
level 1/.style={sibling distance=8mm},
level 2/.style={sibling distance=4mm}, 
level 3/.style={sibling distance=4mm},
] 
\node {$c$}
child  {node  {$a$}}
child  {node  {$a$}}
child  {node  {$a$}}
child  {node  {$a$}}
child  {node  {$b$}
  child {node {$b$}}
  child {node {$b$}}};
\end{tikzpicture}
\caption{Representation of unranked tree $t_1$.}
\label{exemple_arbre_1}
\end{figure}

Another example is $t_2: \{ \Root, 1, 11, 12, 121, 122, 2, 3, 31, 32, 33, 34 \} 
\rightarrow \{ a, b, c \}$ 
as illustrated in \Figure~\ref{exemple_arbre_2}.

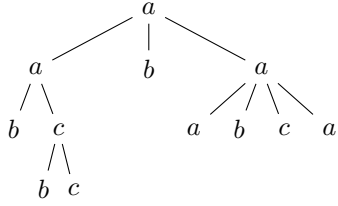
\begin{figure}[h!]
\centering
\begin{tikzpicture}
[level distance=8mm, 
level 1/.style={sibling distance=15mm},
level 2/.style={sibling distance=6mm}, 
level 3/.style={sibling distance=4mm},
] 
\node {$a$}
child  {node  {$a$}
  child {node {$b$}}
  child {node {$c$}
    child {node {$b$}}
    child {node {$c$}}}}
child  {node  {$b$}}
child {node {$a$}
  child {node {$a$}}
  child {node {$b$}}
  child {node {$c$}}
  child {node {$a$}}};
\end{tikzpicture}
\caption{Representation of $t_2$.}
\label{exemple_arbre_2}
\end{figure}
\end{exemple}

\paragraph{Linearization}
Trees can be described by well-balanced 
words which correspond to a depth-first traversal of the tree.
An opening tag is used to notice the arrival on a node
and a closing tag to notice the departure of a node.
For each $a \in \Sigma$, let $a$ itself represent 
the opening tag and $\overline{a}$ the related closing tag.
The \emph{linearization} $\lin{t}$ of $t\in\T\Sigma$ is the 
\emph{well-balanced} word over $\Sigma \cup \overline{\Sigma}$,
with $\overline{\Sigma} = \{ \overline{a}\ |\ a \in \Sigma \}$, 
inductively defined by:
$$ \lin{t} = 
a ~ \lin{\subtree{t}{1}} \cdots \lin{\subtree{t}{n}} ~ \overline{a} $$ 
with $a = t(\epsilon)$ and the root has $n$ children.
We denote by  $\lin{\TSigma}$ 
the set of linearizations of all trees in $\TSigma$.
Let $\PPref{\TSigma}$ denote the set of all proper prefixes of
$\lin{\TSigma}$:
$\PPref{\TSigma} = 
\{ u\in(\Sigma\cup\overline\Sigma)^* \hmid 
\exists v\in(\Sigma\cup\ol\Sigma)^+,\ uv\in\lin{\TSigma}\}$.

\begin{exemple}
Let $t_1$ and $t_2$ be the trees defined in Example
\ref{def_ex_arbres}, then

$\lin{t_1}$ = $c\ a\ \overline{a}$ $a\ \overline{a}$ $a\ \overline{a}$
$a\ \overline{a}$
$b\ b\ \overline{b}$ $b\ \overline{b}\ \overline{b}\ \overline{c}$

$\lin{t_2}$ = $a\ a\ b\ \overline{b}$ $c\ b\ \overline{b}$ 
$c\ \overline{c}\ \overline{c}\ \overline{a}$ 
$b\ \overline{b}$
$a\ a\ \overline{a}$ $b\ \overline{b}$ $c\ \overline{c}$ 
$a\ \overline{a}\ \overline{a}\ \overline{a}$ 
\end{exemple}

\subsection{Visibly pushdown automata} \label{sec:VPA}

Visibly pushdown automata (\VPAs, \cite{AlurMadhusudan04,AlurMadhusudan09}) 
are pushdown automata operating on a partitioned alphabet
where only call symbols can push, return symbols can pop, and internal
symbols can do transitions without considering the stack.

In this paper we only consider languages of unranked trees, 
so we use \VPAs as unranked trees acceptors,
operating on their linearization 
(also named \emph{streaming tree automata} \cite{GauwinNiehrenRoos08}).
This corresponds to the following restrictions.
First, the alphabet is only partitioned into call symbols $\Sigma$
and return symbols $\overline\Sigma$, and does not contain internal symbols.
Second, all linearizations recognized by these \VPAs are such that 
all pairs of matched call $a$ and return $\overline b$ are such that $a=b$,
corresponding to the label of the tree of the corresponding node.
Third, all linearizations are well-matched and single-rooted, 
so the acceptance condition is that a final state is reached on empty stack.

\begin{definition} \label{def:VPA}
A \emph{visibly pushdown automaton} $\AutA$ over alphabet
$\Sigma$ is a tuple $\AutA = (Q, \Sigma, \Gamma, Q_i, Q_f, \Delta)$ where
$Q$ is a finite set of states containing initial states $Q_i\subseteq Q$
and final states $Q_f \subseteq Q $, a finite set $\Gamma$ of stack symbols,
and a finite set $\Delta$ of rules.
Each rule in $\Delta$ is of the form
$\vparule{q}{a}{\gamma}{q'}$
with $a\in\Sigma\cup\overline\Sigma$, $q,q'\in Q$, and $\gamma\in\Gamma$.
\end{definition}

The left-hand side of a rule $\vparule{q}{a}{\gamma}{p}\in\Delta$ is
$(q,a)$ if $a\in\Sigma$, and 
$(q,a,\gamma)$ if $a\in\ol\Sigma$. 
A \VPA is \emph{deterministic} if
it has at most one initial state,
and it does not have two distinct rules with the same left-hand side.

A \emph{configuration} of a \VPA $\AutA$ is a pair $(q,\stack)$
where $q\in Q$ is a state and $\stack\in\Gamma^*$ a stack content. 
A configuration is \emph{initial} (resp. \emph{final}) if $q\in Q_i$ (resp. $q\in Q_f$)
and $\stack=\epsilon$.
For $a\in\Sigma\cup\overline\Sigma$, 
we write $(q,\stack)\xrightarrow{a} (q',\stack')$
if there is a transition $\vparule{q}{a}{\gamma}{q'}$ in $\Delta$
verifying $\stack' = \gamma\cdot\stack$ if $a\in\Sigma$,
and $\stack = \gamma\cdot\stack'$ if $a\in\overline\Sigma$.
We extend this notation to words, by writing 
$(q_0,\stack_0)\xrightarrow{a_1\cdots a_n} (q_n,\stack_n)$
whenever there exist configurations $(q_i,\stack_i)$ such that
$(q_{i-1},\stack_{i-1})\xrightarrow{a_i} (q_i,\stack_i)$ for all $1\le i\le n$.
From $u\in(\Sigma\cup\overline\Sigma)^*$ and 
the set of configurations $\ConfSet\subseteq Q\times\Gamma^*$,
we also define $\Post{u}{\ConfSet}$
as the set of configurations $(q',\stack')$ for which there exists
a configuration $(q,\stack)\in\ConfSet$ such that 
$(q,\stack) \xrightarrow{u} (q',\stack')$.

A \emph{run} of a \VPA $\AutA$ on a linearization 
$\lin t=a_1\cdots a_n$ of $t\in\T\Sigma$
is a sequence $(q_0,\stack_0)\cdots (q_n,\stack_n)$
of configurations $(q_i,\stack_i)$ such that 
$(q_0,\stack_0)$ is initial,
and for every $1\le i\le n$,
$(q_{i-1},\stack_{i-1}) \xrightarrow{a_i} (q_i,\stack_i)$.
Such a run is \emph{accepting} if $(q_n,\stack_n)$ is final.
A tree $t\in\T\Sigma$ is \emph{accepted} by $\AutA$ if there is an accepting run on its
linearization $\lin t$.
The set of accepted trees is called the \emph{language} of $\AutA$ and is written
$\La\AutA$.

\subsection{Universality and $u$-universality}
\label{sec:universality}

We conclude the preliminaries with the notions of universality and $u$-universality,
that we will study in the remainder of the paper. 

\begin{definition}
A tree automaton $\AutA$ over $\Sigma$
is said \emph{universal} if ${\AutA}$ accepts all trees  $t \in \T{\Sigma}$.
Let $u \neq \epsilon$ be a prefix of $\lin{t_0}$
for some tree $t_0\in\T\Sigma$. 
The tree automaton $\AutA$ is said \emph{$u$-universal} 
if for all trees $t \in \T{\Sigma}$, if $u$ is a prefix of $\lin{t}$, then 
$t$ is accepted by  ${\AutA}$.
\end{definition}

In other words, $u$-universality allows to assert that any tree linearization
beginning with $u$ is accepted by the automaton. The two previous definitions does not depend on the tree automaton $\AutA$ but only on the language $\La\AutA$. Therefore they are independent on the kind of tree automata that are used, as soon as they are equivalent. 

Our objective is to propose \emph{incremental} algorithms for $u$-universality,
in the following sense. 
The linearization $\lin{t_0}$ of a given tree $t_0$ is read letter by letter, 
and while $\AutA$ is not $u$-universal for 
the current read prefix $u$ of $\lin{t_0}$, 
the next letter of $\lin{t_0}$ is read. 
For instance Algorithm~\ref{algo:incremental-membership} shows how $u$-universality
is checked incrementally.
When processing a new letter, we try to reuse
prior computations as much as possible.
The automaton can be supposed to be not universal, 
otherwise it is $u$-universal for all such words $u$.

\begin{algorithm}
\begin{algorithmic}
\Function{Incremental-$u$-universality}{$\AutA$, $w$}
\State $i \gets 1$
\While{$i \le |w|$}
  \If {$\AutA$ is $w_1\cdots w_i$-universal}
    \State \Return True
  \EndIf
  \State $i \gets i+1$
\EndWhile
\State \Return False
\EndFunction
\end{algorithmic}
\caption{Checking $u$-universality incrementally
\label{algo:incremental-membership}}
\end{algorithm}

It has been shown in \cite{GauwinNiehrenTison09b} that 
$u$-universality is EXPTIME-complete for \VPAs, 
but in PTIME for deterministic \VPAs.
Determinization is in exponential time for \VPAs,
and our algorithms aim at avoiding this exponential blowup.

An incremental $u$-universality check as described in Algorithm~\ref{algo:incremental-membership} is very useful. First, given a tree $t_0$, it allows a streaming membership test of $t_0$ in $\AutA$: its linearization $\lin{t_0}$ is read letter by letter, and the algorithm declares as soon as possible whether $t_0$ is accepted by $\AutA$. Second, when a property (of XML documents for instance) is given by a tree automaton, then Algorithm~\ref{algo:incremental-membership} detects at the earliest position of $\lin{t_0}$ whether $t_0$ satisfies the property.

\section{Hedge automata approach}
\label{sec:hedge-approach}

We present algorithms for testing universality and $u$-universality 
of a non deterministic visibly pushdown automaton. The approach followed in this section
is based on a translation of the VPA into an hedge automaton. Algorithms with
several optimizations are then provided for checking universality and $u$-universality of hedge automata. 

\subsection{Hedge automata}

We present the standard notion of hedge automata 
\cite{BruggemannkleinMurataWood01,tata},
the usual automaton model for expressing properties on \XML documents.
Indeed, a hedge automaton resembles a DTD:
a DTD is a set of rules like $a \rightarrow b^+c$ saying that 
children of an $a$-node
must be a non empty sequence of $b$-nodes followed by a $c$-node.
Hedge automata are a bit more expressive than DTDs, in that
regular languages operate on states instead of labels,
enabling for instance to distinguish two kinds of $a$-nodes.

A \emph{hedge} $h$ over a finite alphabet $\Sigma$ is a sequence (empty or not)
of unranked trees over $\Sigma$.
The set of all hedges over $\Sigma$ is denoted by $\HdgSigma$.
For instance, given the trees
$t_1$ and $t_2$ from Example~\ref{def_ex_arbres}, the
sequence $t_1t_2t_1$ is a hedge.

\begin{definition}
A \emph{hedge automaton} over $\Sigma$
is a tuple $\AutA = (Q, \Sigma, Q_f, \Delta)$ where 
$Q$ is a finite set of states,
$Q_f \subseteq Q$ is the set of final states,
and $\Delta$ is a finite set of transition rules of  the following type:
$$(a, L, q)$$
where $a \in \Sigma$, $q \in Q$, and 
$L \subseteq Q^*$  is a regular  language over $Q$,
called a {\em horizontal language}.
\end{definition}

We denote by $\LhA$ the set of all horizontal languages of $\AutA$.
Note that  for every $a\in\Sigma$ and $q\in Q$,
we can assume that there is only one $L$ such that $(a,L,q)\in\Delta$.
Indeed, we can replace all rules $(a,L',q)$
by one rule $(a,L,q)$ where $L$ is the union of all such $L'$. 
A hedge automaton is \emph{deterministic} if for all pairs of rules
$(a,L_1,q_1)$ and $(a,L_2,q_2)$ we have $L_1\cap L_2=\emptyset$ or $q_1=q_2$.

A \emph{run} of $\AutA$ on a tree $t \in \T{\Sigma}$ is a tree $\run \in
\T{Q}$ with the same domain as $t$ such that for each node $p \in \Pos{r}$ and 
its $n$ children $p1, p2, \dots, pn$,
if $a = t(p)$ and $q = \run(p)$, then there is a rule $(a,L,q) \in \Delta$ 
with $\run(p1) \run(p2) \dots \run(pn) \in L$.
In particular,
to apply the rule $(a,L,q)$ at a leaf, the empty word $\epsilon$ has to belong to $L$.
Intuitively, a hedge automaton $\AutA$ operates in a bottom-up manner on a tree $t$: 
with a run $\run$, it assigns a state to each leaf, 
and then to each internal node, 
according to the states assigned to its children. 
We use notation
$t \racineRun{\AutA} q$ to indicate the existence of a run $\run$ on $t$ that 
labels the root of $t$ by the state $q$. Such a run $\run$ is \emph{accepting} if
$q$ is final, i.e. $\run(\Root)\in Q_f$. 
An unranked tree~$t$ is \emph{accepted} by $\AutA$ if there exists 
an accepting run on it.
The \emph{language}~$\La{\AutA}$ of $\AutA$ is the set of all unranked trees
accepted by $\AutA$.

\begin{exemple}
\label{exemple_aut}
Let $\AutA = (Q, \Sigma, Q_f, \Delta)$ be a hedge automaton over $\Sigma = \{ a, b, c \}$ with
$Q = \{ q_a, q_b, q_c, q_f \}$, $Q_f = \{ q_f \}$, and $\Delta = \{ 
(a, L_1, q_a),$ $(b, L_1, q_b),$ $(c, L_1, q_c),$
$(a, L_2, q_f),$
$(a, L_3, q_f),$ $(b, L_3, q_f),$ $(c, L_3, q_f) \}$ where 
$L_1 = Q^*$, $L_2 = q_bq_c$ and $L_3 = Q^*q_fQ^*$.

Let $t_1$ and $t_2$ the trees from Example \ref{def_ex_arbres}.
\Figure~\ref{fig_ex_runs} represents
a run $r_1$ of $\AutA$ on $t_1$ and two runs, $r_2$ and $r_3$,
of $\AutA$ on $t_2$. The runs $r_1$ and $r_2$ are not accepting, whereas $r_3$ is accepting.
The tree $t_1$ is not accepted by $\AutA$, whereas  $t_2$ is accepted by $\AutA$.
The language of $\AutA$ is the set of all trees having 
a subtree $s$ which root is an $a$-node and has
two children with $s(1) = b$ and $s(2) = c$.

\begin{figure}[h!]
\centering
\subfloat[Run $r_1$ of $\AutA$ on $t_1$]
{
\begin{tikzpicture}
[level distance=8mm, 
level 1/.style={sibling distance=8mm},
level 2/.style={sibling distance=4mm}, 
level 3/.style={sibling distance=4mm},
] 
\node {$q_c$}
child  {node  {$q_a$}}
child  {node  {$q_a$}}
child  {node  {$q_a$}}
child  {node  {$q_a$}}
child  {node  {$q_b$}
  child {node {$q_b$}}
  child {node {$q_b$}}};
\end{tikzpicture}
}
\hspace{10px}
\subfloat[Run $r_2$ of $\AutA$ on $t_2$]
{
\begin{tikzpicture}
[level distance=8mm, 
level 1/.style={sibling distance=15mm},
level 2/.style={sibling distance=6mm}, 
level 3/.style={sibling distance=4mm},
]
\node {$q_a$}
child  {node  {$q_a$}
  child {node {$q_b$}}
  child {node {$q_c$}
    child {node {$q_b$}}
    child {node {$q_c$}}}}
child  {node  {$q_b$}}
child {node {$q_a$}
  child {node {$q_a$}}
  child {node {$q_b$}}
  child {node {$q_c$}}
  child {node {$q_a$}}};
\end{tikzpicture}
}
\hspace{10px}
\subfloat[Run $r_3$ of $\AutA$  on $t_2$]
{
\begin{tikzpicture}
[level distance=8mm, 
level 1/.style={sibling distance=15mm},
level 2/.style={sibling distance=6mm}, 
level 3/.style={sibling distance=4mm},
] 
\node {$q_f$}
child  {node  {$q_f$}
  child {node {$q_b$}}
  child {node {$q_c$}
    child {node {$q_b$}}
    child {node {$q_c$}}}}
child  {node  {$q_b$}}
child {node {$q_a$}
  child {node {$q_a$}}
  child {node {$q_b$}}
  child {node {$q_c$}}
  child {node {$q_a$}}};
\end{tikzpicture}
}
\caption{Examples of runs}
\label{fig_ex_runs}
\end{figure}
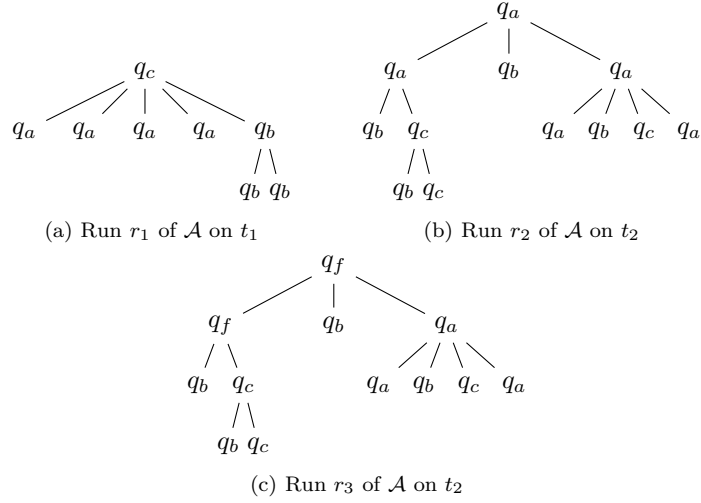
\end{exemple}

\subsection{From VPAs to hedge automata} 

In this section, we describe a translation of \VPAs into hedge automata, 
with the aim to transfer universality and $u$-universality testing of a \VPA 
to a hedge automaton.

\begin{theorem} \label{thm:VPAVersHedgeaut}
Let $\AutA$ be a \VPA. Then one can construct a hedge automaton $\AutA_H$ such that for all $t \in \TSigma$, $\lin{t} \in \La{\AutA}$ iff $t \in \La{\AutA_H}$.
\end{theorem}
 
\begin{proof}
Let $\AutA = (Q, \Sigma, \Gamma, Q_i, Q_f, \Delta)$ be a \VPA.
We define the hedge automaton $\AutA_{H} = (Q', \Sigma, Q'_f, \Delta')$
such that 
\begin{itemize}
\item $Q'  = Q \times Q$
\item $Q'_f = Q_i \times Q_f$
\item
$\Delta' = \{ (a, L_{s,s'}, (q,q')) \mid \exists \gamma \in \Gamma, 
\vparule{q}{a}{\gamma}{s} \in \Delta
\text{ and } 
\vparule{s'}{\ol{a}}{\gamma}{q'} \in \Delta  \}$
where $L_{s,s'} = \{(s,q_1)\cdot (q_1,q_2)\cdots (q_{n-1},q_n)\cdot (q_n,s') \mid n \geq 0,\ 
s,q_1, \dots, q_n, s' \in Q \} \cup K_{s,s'}$, and  $K_{s,s'} = \emptyset$ if $s \neq s'$,  $K_{s,s'} = \{\epsilon\}$ otherwise.
\end{itemize}

Notice that each language $L_{s,s'}$ is regular.
Let us prove for all $t \in \TSigma$ and $q,q' \in Q$ that:
$$
(q,\epsilon) \xrightarrow{\lin{t}} (q',\epsilon) 
~\iff~
t  \racineRun{\AutA_H} (q,q')
$$
As a consequence, we will have $\lin{t} \in \La{\AutA}$ iff $t \in \La{\AutA_H}$.

We proceed by induction on the height of $t$.
We begin with the basic case $\height{t} = 1$, i.e.  $t$ be a $a$-leaf for some $a \in \Sigma$.
Then $t  \racineRun{\AutA_H} (q,q')$ iff $\exists s \in Q, \gamma \in \Gamma$ such that $\vparule{q}{a}{\gamma}{s}\in \Delta$ and $\vparule{s}{\ol{a}}{\gamma}{q'}  \in \Delta$ (recall that $\epsilon \in L_{s,s}$). This is equivalent to $(q,\epsilon) \xrightarrow{\lin{t}} (q',\epsilon)$.

Let $i > 1$ and suppose that the property holds for all trees of height less than~$i$. Let $t$ be a tree of height $i$ such that $a = t(\epsilon)$ and the root has $n$ children.

Let $r$ be a run of $\AutA_H$ on $t$ such that $(q,q') = r(\epsilon)$. 
Then, by definition of $\AutA_H$, there exist $q_1, \cdots, q_{n+1} \in Q$ 
and $\gamma \in \Gamma$ such that 
$r(1) = (q_1,q_2)$, $r(2) = (q_2,q_3), \ldots, 
r(n) = (q_n,q_{n+1})$, $\vparule{q}{a}{\gamma}{q_0} \in \Delta$ 
and $\vparule{q_{n+1}} {\ol{a}}{\gamma}{q'} \in \Delta$.
We know by induction hypothesis that 
$(q_i,\epsilon)  \xrightarrow{\lin{\subtree{t}{i}}} (q_{i+1},\epsilon)$ 
for all $1 \le i \le n$. 
It follows that $(q_0,\epsilon)  \xrightarrow{\lin{h}} (q_{n+1},\epsilon)$ 
where $h = \subtree{t}{1}\subtree{t}{2}\cdots\subtree{t}{n}$. 
We have also $(q_0,\gamma)  \xrightarrow{\lin{h}} (q_{n+1},\gamma)$ 
since $h$ is an edge, and thus 
$(q,\epsilon) \xrightarrow{\lin{t} = a\lin{h}\overline{a}}  (q',\epsilon)$.

Suppose now that $(q,\epsilon) \xrightarrow{\lin{t}}  (q',\epsilon)$. So there exist $q_1, \dots, q_{n+1} \in Q$ and $\gamma \in \Gamma$ such that $\vparule{q}{a}{\gamma}{q_1}  \in \Delta$, $\vparule{q_{n+1}}{\ol{a}}{\gamma}{q'}  \in \Delta$, and $(q_i,\gamma)  \xrightarrow{\lin{\subtree{t}{i}}} (q_{i+1},\gamma)$ for all $i$.
By induction hypothesis, $\subtree{t}{i}  \racineRun{\AutA_H} (q_i,q_{i+1})$ for all $i$, and thus $t \racineRun{\AutA_H} (q,q')$.
\end{proof}

As a consequence of Theorem~\ref{thm:VPAVersHedgeaut}, universality and $u$-universality testing of a \VPA $\AutA$ is transfered to the hedge automaton $\AutA_H$.

\subsection{Checking universality}

A standard method to check universality of a hedge automaton
is to determinize it, complement it, and check for emptiness.
As determinization is in exponential time \cite{tata}, 
we propose in this section
an antichain-based algorithm for checking universality
without explicit determinization.

Such an algorithm has been proposed in \cite{BouajjaniHHTV08}
for finite (ranked) tree automata. In the context of hedge automata,
additional difficulties have to be solved due to the fact that
the accepted trees are unranked.

In our approach, the main idea is to find as fast as possible
one tree rejected by the hedge automaton (if it exists)
by performing a kind of bottom-up implicit determinization.
Antichains will limit the computations.

\subsubsection{Macrostates and $\PostOp$ operator} \label{sec:Post}

To test universality of a hedge automaton $\AutA$,
we have to check that all the trees of $\T{\Sigma}$  
belong to $\La{\AutA}$. 
Instead of working with trees we work with sets of
states, which are called \emph{macrostates}.
A macrostate is associated with each tree~$t$: it is 
the set of all the states $q$ labeling the root of a run of 
$\AutA$ on $t$, i.e. such that $t \racineRun{\AutA} q$. 
To compute the macrostates,
we make bottom-up computations by applying
a $\PostOp$ operator defined as follows.

\begin{definition}
Let $\AutAGeneral$ be a hedge automaton.
A \emph{macrostate} is a set of states $P  \subseteq Q$. A 
\emph{macrostate word} $\pi = P_1 P_2 \cdots P_n$, $n \ge 0$, is 
a word over the alphabet $2^Q$. We denote by $\wordsOf{\pi}$
the set $\{p_1p_2 \cdots p_n \mid p_i \in P_i, \forall i, 1 \leq i \leq n\}$.
Given $a \in \Sigma$ and $\pi$ a macrostate word, let
$$\Post{a}{\pi} =  \{q \in Q \mid \exists (a,L,q) 
\in \Delta: L \cap \wordsOf{\pi} \neq \emptyset \}$$
For $\macroSet \subseteq 2^Q$ a set of macrostates, let
$$\PostAll{\macroSet} = \{ \Post{a}{\pi} \mid a \in \Sigma, \pi \in
\macroSet^* \} \cup \macroSet$$
and $\PostAllStar{\macroSet} = \cup_{i \ge 0} 
\PostAlli{i}{\macroSet}$
such that $\PostAlli{0}{\macroSet} = \macroSet$, and for all
$i > 0$, $\PostAlli{i}{\macroSet} = \PostAll{\PostAlli{i-1}{\macroSet}}$.
\end{definition}

When $\pi = \epsilon$, $\Post{a}{\epsilon}$ is the set of all
states that can be assigned to an $a$-leaf of a tree, with $a \in \Sigma$. 
If an $a$-node has $n$ children to which the macrostates $P_1$, $\dots,$
$P_n$ have been assigned, then $\Post{a}{P_1 \cdots P_n}$ is the set
of all states that can be assigned to this node. The next lemma is immediate.

\begin{lemma}
\label{lem_macrostate_Pi}
Let $\AutAGeneral$ be a hedge automaton and $t \in \TSigma$ be such that
its root is an $a$-node with $n$ children.
Let $P_i = \{ q \in Q \mid \subtree{t}{i} \racineRun{\AutA} q \}$
for $1\le i\le n$.
Then 
$$\Post{a}{P_1 \cdots P_n} = \{ q \in Q \mid  t \racineRun{\AutA} q \}.$$
\end{lemma}

Given $\macroSet$ a set of macrostates, $\PostAll{\macroSet}$ is the
set of all macrostates that belong to $\macroSet$ or can be obtained via
$\Post{a}{\pi}$ with any letter $a \in \Sigma$, and any macrostate word
$\pi = P_1 P_2 \cdots P_n$ with $P_i \in \macroSet, \forall i$. More precisely, we have:

\begin{lemma} \label{lem:pointfixe}
\label{postall_i_emptyset}
Let $\AutAGeneral$ be a hedge automaton and $i \ge 1$.
A macrostate $P$ belongs to $\PostAlli{i}{\emptyset}$ iff
there exists a tree $t \in \TSigma$ with $\height{t} \le i$ such that
$P = \{ q \in Q \mid t \racineRun{\AutA} q \}$.
\end{lemma}

\begin{proof}
We proceed by induction on $i$.

The basic case, $i = 1$, directly follows from
$\PostAlli{1}{\emptyset} = \{ \Post{a}{\epsilon} 
\mid a \in \Sigma \}$ and 
$\Post{a}{\epsilon} = \{ q \mid t \racineRun{\AutA} q \}$ with
$t$ being an $a$-leaf.

Let $i > 1$
and suppose that the property holds for all $j, 1 \le j < i$.

\noindent
($\Rightarrow$) Let $P \in \PostAlli{i}{\emptyset}$. If $P \in \PostAlli{i-1}{\emptyset}$,
then the property holds by induction hypothesis. Otherwise there exist $n \ge 0$,
$P_1, \dots, P_n \in \PostAlli{i-1}{\emptyset}$, and  $a \in \Sigma$,
such that $P = \Post{a}{P_1 \cdots P_n}$.
By induction hypothesis, $\forall k, 1 \le k \le n$, $\exists t_k \in \TSigma$
such that $\height{t_k} < i$ and 
$P_k = \{ q \mid t_k \racineRun{\AutA} q \}$. 
Let $t$ be the tree with the $a$-root and the $n$ subtrees $t_1, \dots, t_n$. Then $\height{t} \le i$
and $P = \{ q \mid t \racineRun{\AutA} q \}$ by Lemma \ref{lem_macrostate_Pi}.

\noindent
($\Leftarrow$) Let $t \in \TSigma$ with $\height{t} \leq i$ 
and $P = \{ q \in Q \mid t \racineRun{\AutA} q \}$. If $\height{t} < i$, then
by induction hypothesis $P \in \PostAlli{i-1}{\emptyset} \subseteq \PostAlli{i}{\emptyset}$. 
Otherwise let $a$ be the label of the root of $t$ and 
$\subtree{t}{1}, \dots, \subtree{t}{n}$ its $n$ subtrees.
Let $P_k = \{ q \in Q \mid \subtree{t}{k} \racineRun{\AutA} q \}$, $1 \le k \le n$.
As $\height{\subtree{t}{k}} < i$, we have by induction hypthesis that $P_k \in
\PostAlli{i-1}{\emptyset}$. By Lemma \ref{lem_macrostate_Pi},
$P = \Post{a}{P_1 \cdots P_n}$, and thus $P \in \PostAlli{i}{\emptyset}$.
\end{proof}

Given a tree $t \in \TSigma$ we define $\Pt{t}$ as the macrostate $\Pt{t} = \{ q \in Q \mid t \racineRun{\AutA} q \}$. 
More generally, given a hedge $h = t_1t_2 \cdots t_n \in \HdgSigma$ 
we denote by $\pih{h}$ the macrostate word 
$\pih{h} = P_{t_1}P_{t_2} \cdots P_{t_n}$.
The previous lemmas indicate that $\PostAllStar{\emptyset} = \{ \Pt{t} \mid t \in \TSigma\}$, and more generally that $(\PostAllStar{\emptyset})^* = \{ \pih{h} \mid h \in \HdgSigma\}$. 

The next proposition is an immediate consequence of Lemmas \ref{lem_macrostate_Pi} and \ref{postall_i_emptyset}.

\begin{proposition}
\label{prop_universality}
Let $\AutAGeneral$ be a hedge automaton. Then
$\AutA$ is universal iff $\forall P \in \PostAllStar{\emptyset}$,
$P \cap Q_f \neq \emptyset$.
\end{proposition}

\subsubsection{Relations and universality algorithm} \label{universality-relations}

Our method for checking universality of a hedge automaton
is to compute $\PostAllStar{\emptyset}$ by
iteratively applying the $\PostOp$ operator. However
to get $\PostAll{\macroSet}$, we have to compute $\macroSet^*$ which is
an infinite set of macrostate words.
To circumvent this problem, we represent a macrostate word by a relation
as described below, with the advantage that the set of relations is now finite.

We first introduce some notation. Let $\AutAGeneral$ be a hedge automaton and 
$\LhA$ be the set of horizontal languages appearing in its transition rules.
We recall that these languages are regular.
Let $L \in \LhA$ and $\AutB_L$ be a (word) automaton over the alphabet $Q$ that accepts $L$.
Let $S_L$ be its set of states,  $I_L$ its set of initial states, and $F_L$ its set of final states.
We denote by $\BLh$ the automaton which is the disjoint union of all 
the automata $\AutB_L$ with $L \in \LhA$. Its set of states is denoted by 
$\SLh = \underset{L \in \LhA}{\bigcup} S_L$. 
A run in $\BLh$ from state $s \in \SLh$ to state $s' \in \SLh$
labeled by word $w \in Q^*$ is denoted by $s \pathLabeled{w} s'$.

\begin{definition}
Let $\AutAGeneral$ be a hedge automaton  and
$\pi$ a macrostate word. Then $\rel{\pi} \subseteq \SLh \times \SLh$ 
is the relation
$$ \rel{\pi} = \{ (s,s') \mid  s \pathLabeled{w} s' 
\text{ with } w \in \wordsOf{\pi} \}.$$
\end{definition}

In other words, if $\pi = P_1 \cdots P_n$ with $P_i \subseteq Q$ for all $i$, 
then $(s,s')$ belongs to $\rel{\pi}$ iff there is a path in $\BLh$ from $s$ to $s'$ 
that is labeled by a word $p_1 \cdots p_n \in \wordsOf{\pi}$. 
The notation $\relOp$ is naturally extended to sets $\mswSet$ of macrostate words as $\rel{\mswSet} = \{\rel{\pi} \mid \pi \in \mswSet\}$.

Notice there are finitely many relations $r \subseteq \SLh \times \SLh$ , since
$\SLh$ is a finite set. If $\relSet$ is a set of relations $r  \subseteq \SLh \times \SLh$, then
$\relSet^*$ denotes the set of all relations obtained by composing 
relations in $\relSet$:
$\relSet^* = \{r_1\circ r_2\circ\cdots \circ r_n \hmid 
n\ge 0 \text{ and } r_i\in\relSet \text{ for all } 1\le i\le n \}$. 
In particular $\relSet^*$ contains the identity relation $\vpaid_{\SLh}$ over $\SLh$,
obtained when $n=0$.

\begin{lemma}
\label{relSetStar-macroSetStar}
Let $\AutAGeneral$ be a hedge automaton. If $\macroSet$ a set of macrostates
and $\relSet$ a set of relations such that $\rel{\macroSet} = \relSet$, then
$\rel{\macroSet^*} = \relSet^*$.
\end{lemma}

\begin{proof}
Let us prove that for any macrostate word $\pi = P_1 \cdots P_n$,
$\rel{\pi} = \rel{P_1} \circ\cdots \circ \rel{P_n}$ ; the lemma is an immediate
consequence.

Let $(s, s') \in \rel{P_1 \cdots P_n}$, that is, $\exists w = p_1 \cdots p_n \in \wordsOf{\pi} :
s \pathLabeled{w} s'$. Let $s = s_1, s_2, \cdots, s_n, s_{n+1} =  s' \in \SLh$ 
be such that $s_i \pathLabeled{p_i} s_{i+1}$ for all $i$.
As $p_i \in P_i$ and $(s_i, s_{i+1}) \in \rel{P_i}$, it follows that $(s,s') \in
\rel{P_1} \circ \cdots \circ \rel{P_n}$.

Conversely, let $(s,s') \in \rel{P_1} \circ \cdots \circ \rel{P_n}$. 
Let $s = s_1, s_2, \cdots, s_n, s_{n+1} =  s' \in \SLh$ be such that
$(s_i,s_{i+1})  \in \rel{P_i}$ for all $i$.
By definition, for all $i$, there exists $p_i \in P_i$ such that  
$s_i \pathLabeled{p_i} s_{i+1}$.
So for $w = p_1 \cdots p_n$, we have  $s_1 \pathLabeled{w} s_{n+1}$ 
showing that $(s, s') \in \rel{P_1 \cdots P_n}$.
\end{proof}

The $\PostOp$ operator is adapted to relations in the following way.

\begin{definition}
Let $\AutAGeneral$ be a hedge automaton, $r \subseteq \SLh \times \SLh$ a 
relation,
and $a \in \Sigma$ a letter. Then
$$\Posta{r} =  \{ q \in Q \mid \exists (a,L,q) 
\in \Delta, \exists (s, s') \in r:
s \in I_L \text{ and } s' \in F_L \}. $$
\end{definition}

\begin{lemma}
\label{post_a_rel_pi}
Let $a \in \Sigma$ 
and $\pi$ be a macrostate word, then $\Posta{\pi} = \Posta{\rel{\pi}}$.
\end{lemma}

\begin{proof}
For  $a \in \Sigma$
and $\pi$ a macrostate word, we have
\begin{eqnarray} \nonumber
\Post{a}{\pi} 
&=&  \{q \in Q \mid \exists (a,L,q) 
\in \Delta: L \cap \wordsOf{\pi} \neq \emptyset \} \\ \nonumber
&=& \{ q \in Q \mid \exists (a,L,q) \in \Delta,
\exists s, s' \in \SLh, \exists w \in \wordsOf{\pi} : s \pathLabeled{w} s', s \in I_L
\text{ and } s' \in F_L \} \\ \nonumber
&=& \{ q \in Q \mid \exists (a,L,q) 
\in \Delta, \exists (s, s') \in \rel{\pi} :
s \in I_L \text{ and } s' \in F_L \} \\ \nonumber
&=& \Post{a}{\rel{\pi}}.
\end{eqnarray}
\end{proof}

\begin{lemma} \label{lem:calcul de post avec rel}
\label{post_all_relSet}
Let $\macroSet$ be a set of macrostates,
then $\PostAll{\macroSet} = \{ \Post{a}{r} \mid a \in \Sigma, r \in
\rel{\macroSet}^* \} \cup \macroSet$.
\end{lemma}

\begin{proof}
By definition, $\PostAll{\macroSet} = 
\{ \Post{a}{\pi} \mid a \in \Sigma, \pi \in
\macroSet^* \}$. By Lemma \ref{post_a_rel_pi}, this set is equal to 
$\{ \Post{a}{\rel{\pi}} \mid a \in \Sigma, \pi
\in \macroSet^* \}$ which is equal to 
$\{ \Post{a}{r} \mid a \in \Sigma, 
r \in \rel{\macroSet}^* \}$ by Lemma \ref{relSetStar-macroSetStar}.
\end{proof}

We are now able to propose an algorithm to check universality of hedge automata.
With Algorithm \ref{algo:universality}, the set $\PostAllStar{\emptyset}$ is computed incrementally 
and the universality test is performed thanks to Proposition~\ref{prop_universality}. 
More precisely, at step $i$, variable $\macroSet$ is used for $\PostAlli{i}{\emptyset}$ 
and variable $\relSet^*$ is used for $\rel{\macroSet}^*$. We compute $\relSet^*$ with Function \textproc{CompositionClosure}, and then possible new macrostates with $\{ \Post{a}{r} \mid a \in  \Sigma, r \in \relSet^* \}$. The algorithm stops when no new macrostate is found or the hedge automaton is declared not universal.

\begin{algorithm}
\begin{algorithmic}
\Function{Universality}{$\AutA$}
\State $\macroSet \gets \emptyset$
\State $\relSet^* \gets \{\vpaid_{\SLh}\}$
\Repeat
    \State $\macroSet_{new} \gets \{ \Post{a}{r} \mid a \in  \Sigma, r \in \relSet^* \}$
    \If {$\exists P \in \macroSet_{new}: P \cap F = \emptyset$}
        \State \Return False \quad // Not universal
    \EndIf
     \State $\relSet' \gets \rel{\macroSet_{new} \setminus \macroSet} \setminus \relSet^*$
     \If  {$\relSet' \neq \emptyset$}
         \State $\macroSet \gets \macroSet \cup \macroSet_{new}$
         \State $\relSet^* \gets \textproc{CompositionClosure}(\relSet^*,\relSet')$
    \EndIf
\Until{$\relSet' = \emptyset$}
\State \Return True \quad // Universal
\EndFunction
\end{algorithmic}
\caption{Checking universality
\label{algo:universality}}
\end{algorithm}

Let us detail Function $\textproc{CompositionClosure}(\relSet^*, \relSet')$ which computes the set  $(\relSet^* \cup \relSet')^*$. In Algorithm~\ref{algo:composition}, we show how to compute $(\relSet^* \cup \relSet')^*$ given the inputs $\relSet^*$ and $\relSet'$, without recomputing $\relSet^*$ from $\relSet$. Initially, \emph{Relations} is equal to $\relSet^* \cup \relSet'$ and will be equal to $(\relSet^* \cup \relSet')^*$ at the end of the computation.  \emph{ToProcess} contains the relations that can produce new relations by composition with an element of \emph{Relations}.

\begin{algorithm}
\newcommand{\allfuncs}{{\it Relations}}
\newcommand{\toprocess}{{\it ToProcess}}
\newcommand{\newfuncs}{{\it NewRelations}}
\newcommand{\compose}{{\it compose}}
\newcommand{\addroot}{{\it add\_root}}
\newcommand{\func}{{\it rel}}
\newcommand{\compositions}{{\it compositions}}

\begin{algorithmic}
\Function{CompositionClosure}{$\relSet^*,\relSet'$}
    \State $\allfuncs \gets \relSet^* \cup \relSet'$
    \State $\toprocess \gets \relSet'$
    \While{$\toprocess \not= \emptyset$}
        \State $\func \gets \textproc{Pop}(\toprocess)$
        \State $\newfuncs \gets \emptyset$
        \For {$r\in\allfuncs$}
            \State $\newfuncs \gets 
            \newfuncs \cup \{ r \circ \func, \func \circ r \}$
        \EndFor
        \State $\toprocess \gets \toprocess \cup (\newfuncs \setminus \allfuncs)$
        \State $\allfuncs \gets \allfuncs \cup \newfuncs$
    \EndWhile
    \State \Return \allfuncs
\EndFunction
\end{algorithmic}
\caption{Computing $(\relSet^* \cup \relSet')^*$
\label{algo:composition}}
\end{algorithm}

\begin{proposition}
Given $\relSet^*$ and $\relSet'$, Algorithm~\ref{algo:composition} computes $(\relSet^* \cup \relSet')^*$.
\end{proposition}

\begin{proof}
Let \emph{Relations} be the set computed by Algorithm~\ref{algo:composition}. Clearly, $\text{\emph{Relations}} \subseteq (\relSet^* \cup \relSet')^*$. Assume by contradiction there exists $r$ that belongs to  $(\relSet^* \cup \relSet')^* \setminus \text{\emph{Relations}}$. Then $r \not\in \relSet^* \cup \relSet'$ and we can suppose wlog that $r = r'_2 \circ r'_1$ with $r'_1, r'_2 \in \text{\emph{Relations}}$. Notice that at least one element among $r'_1, r'_2$ has been added to \emph{ToProcess} during the execution of Algorithm~\ref{algo:composition}, since otherwise $r'_1, r'_2 \in \relSet^*$ and then $r \in \relSet^*$. If $r'_1$ is the last one (among $r'_1, r'_2$) to be popped from \emph{ToProcess}, then the relation $r'_2 \circ r'_1$ is added to \emph{NewRelations}, which leads to a contradiction. The conclusion is similar if $r'_2$ is is the last one to be popped.
\end{proof}

\subsubsection{Antichain-based optimization} \label{sec:antichain-univ}

In this section we explain how to use the
concept of antichain for saving computations.
We show that it is sufficient to only compute
the $\subseteq$-minimal elements of $\PostAllStar{\emptyset}$
for checking universality.

Consider the set $2^Q$ of all macrostates, with the $\subseteq$ operator. 
An \emph{antichain} $\macroSet$ of macrostates is a set of 
pairwise incomparable macrostates with respect to $\subseteq$. 
Given a set $\macroSet$ of macrostates, we denote by 
$\lfloor \macroSet \rfloor$ the \emph{$\subseteq$-minimal} elements of 
$\macroSet$, similarly we denote by $\lceil \macroSet \rceil$ 
the \emph{$\subseteq$-maximal} elements of $\macroSet$. 
A set $\macroSet$ of macrostates is \emph{$\subseteq$-upward closed} 
(resp. \emph{$\subseteq$-downward closed}) 
if for all $P \in \macroSet$ and $P \subseteq P'$ (resp. $P' \subseteq P$), 
we have $P' \in \macroSet$.
The same notions can be defined for a set of relations (instead of macrostates). 

\begin{definition}
Let $\AutAGeneral$ be a hedge automaton. Let $\macroSet \subseteq 2^Q$ be a set of macrostates, let 
$$\MinPostAll{\macroSet} =  \lfloor \PostAll{\macroSet} \rfloor$$
and $\MinPostAllStar{\macroSet} = \cup_{i \ge 0} 
\MinPostAlli{i}{\macroSet}$
such that $\MinPostAlli{0}{\macroSet} =  \lfloor \macroSet \rfloor$, and for all
$i > 0$, $\MinPostAlli{i}{\macroSet} = \MinPostAll{\MinPostAlli{i-1}
{\macroSet}}$.
\end{definition}

\begin{lemma}
\label{lemma_postallstar_and_subset_order}
Given $\macroSet$ a set of macrostates, for all $P \in \PostAllStar{\macroSet}$, there exists $P' \in
\MinPostAllStar{\macroSet}$ such that $P' \subset P$.
\end{lemma}

\begin{proof}
The proof is done by induction on $i$ such that $\PostAllStar{\macroSet} = \cup_{i \ge 0} 
\PostAlli{i}{\macroSet}$, and on the next two observations:
\begin{itemize}
\item Given $a \in \Sigma$, and $r, r'$ two relations over $\SLh$, if 
$r \subseteq r'$ then $\Post{a}{r} \subseteq \Post{a}{r'}$.
\item Let $r_1, \cdots, r_n, r_1',
\cdots, r_n'$ be relations over $\SLh$, if  $r_i \subseteq r_i', \forall 1 \le i \le n$, then
$r_1 \circ \cdots \circ r_n \subseteq r_1' \circ \cdots \circ r_n'$.
\end{itemize}
\end{proof}

Notice that thanks to Lemma~\ref{lem:calcul de post avec rel}, given an antichain of macrostates $\macroSet$, we can compute $\MinPostAll{\macroSet}$ as $\lfloor \{ \Post{a}{r} \mid a \in \Sigma, r \in \lfloor \rel{\macroSet}^* \rfloor \} \cup \macroSet \rfloor.$ 
We have the next counterpart of Proposition~\ref{prop_universality}.

\begin{proposition}
Let $\AutAGeneral$ be a hedge automaton.
$\AutA$ is universal if and only if $\forall P \in 
\MinPostAllStar{\emptyset}, P \cap Q_f \neq \emptyset$.
\end{proposition}

\begin{proof}
The proof is based on Proposition~\ref{prop_universality}.

\noindent
($\Rightarrow$) As $\MinPostAllStar{\emptyset} \subseteq \PostAllStar{\emptyset}$,
the proof is immediate.

\noindent
($\Leftarrow$) Suppose that $\forall P \in  \MinPostAllStar{\emptyset}, 
P \cap Q_f \neq \emptyset$. Let $P' \in \PostAllStar{\emptyset}$.
By Lemma~\ref{lemma_postallstar_and_subset_order}, $\exists P \in
\MinPostAllStar{\emptyset} : P \subseteq P'$. It follows that $P' \cap Q_f \neq \emptyset$.
\end{proof}

Algorithm \ref{algo:universality_antichain} checks whether a given
hedge automaton is universal by computing incrementally
$\MinPostAllStar{\emptyset}$. It is an adaptation of 
Algorithm \ref{algo:universality}. 

\begin{algorithm}
\begin{algorithmic}
\Function{Universality}{$\AutA$}
\State $\macroSet \gets \emptyset$
\State $\relSet^*_{min} \gets \{\vpaid_{\SLh}\}$
\Repeat
    \State $\macroSet_{new} \gets \lfloor \{ \Post{a}{r} \mid a \in  \Sigma, r \in \relSet^*_{min} \} \rfloor$
    \If {$\exists P \in \macroSet_{new}: P \cap F = \emptyset$}
        \State \Return False \quad // Not universal
    \EndIf
     \State $\relSet' \gets \rel{\macroSet_{new} \setminus \macroSet} \setminus \relSet^*_{min}$
     \If  {$\relSet' \neq \emptyset$}
         \State $\macroSet \gets \lfloor \macroSet \cup \macroSet_{new} \rfloor$
         \State $\relSet^*_{min} \gets \lfloor \textproc{CompositionClosure}(\relSet^*_{min},\relSet') \rfloor$
    \EndIf
\Until{$\relSet' = \emptyset$}
\State \Return True \quad // Universal
\EndFunction
\end{algorithmic}
\caption{Checking universality
\label{algo:universality_antichain}}
\end{algorithm}

Notice that in Algorithm \ref{algo:universality_antichain}, to compute 
$\lfloor \rel{\macroSet}^* \rfloor$, 
we first make a call to Function \textproc{CompositionClosure} and then we only keep
the $\subseteq$-minimal elements of the result. An optimisation could be,
at each step of the \textproc{CompositionClosure} computation, to only consider the
minimal elements.

\subsection{Checking $u$-universality} \label{sec:HedgeAutomataAndu-universality}

In this section, given $\AutA$ a hedge automaton and 
$u \neq \epsilon$ a word in $\PPref{\TSigma}$, 
we propose a method to check whether $\AutA$ is $u$-universal. 
This method is incremental, as explained in Section~\ref{sec:universality}.
As in the previous section, we first propose our approach, then transform it into an algorithm (thanks to relations), and finally propose some optimizations.

We need the following notation. Let $u$ be the current read proper prefix of $\lin{t_0}$ for a given tree $t_0$. If $u = a_1\lin{h_1}a_2\lin{h_2} \cdots a_n\lin{h_n}$ with $a_i \in \Sigma, h_i \in \HdgSigma$, for $1 \leq i \leq n$, then $\open{u} =  a_1a_2 \cdots a_n$. 
In other words, $a_1, a_2, \ldots, a_n$ are the read open tags 
which closing tags have not been read yet. 
The partial reading of $t_0$ according to $u$ indicates a current list of ancestors respectively labeled by $a_1, a_2, \ldots a_n$ as depicted in \Figure \ref{open_u}.

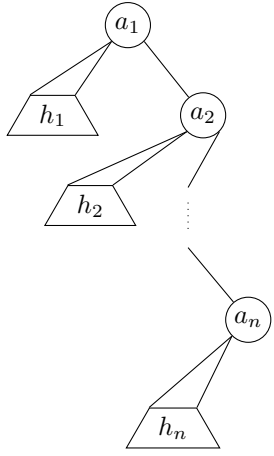
\begin{figure}[h!]
\centering
\begin{tikzpicture}
\def\step{1.2}
\def\istep{1.8*\step-0.2}
\def\dstep{\istep-0.4}
\def\nstep{\dstep-0.2}
\def\msize{6mm}

\node at (0,0) {$a_1$};
\node[circle, draw, minimum size=\msize] (a1) at (0,0) {};
\node[trapezium, draw] (h1) at (-1,-\step) {$h_1$};
\node at (1,-\step) {$a_2$};
\node[circle, draw, minimum size=\msize] (a2) at (1,-\step) {};
\node[trapezium, draw] (h2) at (-0.5,-2*\step) {$h_2$};
\coordinate (a3) at (0.8,-1.8*\step);
\coordinate (ai) at (0.8, -\istep);
\coordinate (ai-plus) at (0.8, -\dstep);
\coordinate (an-moins) at (0.8, -\nstep);
\node at (1.6,  -\nstep-0.8*\step) {$a_n$};
\node[circle, draw, minimum size=\msize] (an) at (1.6,  -\nstep-0.8*\step) {};
\node[trapezium, draw]  (hn) at (0.6, -\nstep-2*\step) {$h_n$};

\draw (a1.south west) -- (h1.top left corner);
\draw (a1.south west) -- (h1.top right corner);
\draw (a1.south east) -- (a2);
\draw (a2.south west) -- (h2.top left corner);
\draw (a2.south west) -- (h2.top right corner);
\draw (a2.south east) -- (a3);
\draw[dotted] (ai) -- (ai-plus);
\draw (an-moins) -- (an);
\draw (an.south west) -- (hn.top left corner);
\draw (an.south west) -- (hn.top right corner);
\end{tikzpicture}
\caption{Current reading of a tree $t_0$ according to the prefix $a_1\lin{h_1}a_2\lin{h_2} \cdots  a_n\lin{h_n}$.}
\label{open_u}
\end{figure}

Given $u$, let $w_i = a_1\lin{h_1} \cdots a_{i-1}\lin{h_{i-1}}$, for 
$1 \leq i \leq n$, such that $w_1 = \epsilon$. The incremental method is based 
on the usage of some sets 
$$\X_{w_ia_i},  ~1 \leq i \leq n,$$ 
such that each 
$\X_{w_ia_i}$ is defined from $\X_{w_{i-1}a_{i-1}}$, with the underlying idea that 
$\AutA$ is $w_ia_i$-universal iff $\X_{w_ia_i}$ is empty.
This permits to check $u$-universality when $u$ ends with a $\Sigma$-symbol.
Moreover, we will see that each element of $\X_{w_ia_i}$
is a witness of some word $v$ such that the tree $t$ with
$\lin{t} = w_ia_iv$ is not accepted by $\AutA$. 
For words $u$ ending with a $\ol\Sigma$-symbol,
we will explain at the end of this section how the test of 
$w_ia_i\lin{h_i}$-universality can be easily performed 
using the set $\X_{w_ia_i}$.

\subsubsection{Incremental approach}

Let us give the definition of $\X_{w_ia_i}$ for all $i$. We begin with the basic case $i=1$, i.e. with set 
$X_a$.

We use notation $\PointFixe$ for $\PostAllStar{\emptyset}$ and $\PointFixeEtoile$ for $(\PostAllStar{\emptyset})^*$ as introduced in Section~\ref{sec:Post} (recall that $\PostAllStar{\emptyset} = \{ \Pt{t} \mid t \in \TSigma\}$ and $(\PostAllStar{\emptyset})^* = \{ \pih{h} \mid h \in \HdgSigma\}$ by Lemma~\ref{lem:pointfixe}).
Given a set $\mswSet$ of macrostate words, 
we define $\Pref{\mswSet}$ as the set 
$\{\pi \in \PointFixeEtoile \mid \exists\pi' \in \PointFixeEtoile :  
\pi  \pi' \in   \mswSet \}$.

\paragraph{Basic case}
We need to define $\X_{a}$ such that $\X_a = \emptyset$ iff $\AutA$ is $a$-universal, i.e. all trees $t$ such that $\lin{t} = a \lin{h}\overline{a}$ with $h \in \HdgSigma$, are accepted by $\AutA$.  The test of $a$-universality is performed in two steps. We first collect all macrostate words $\pih{h} \in\PointFixeEtoile$ (see Lemmas~\ref{lem_macrostate_Pi} and~\ref{postall_i_emptyset}). 
Then for each of them we compute $\Post{a}{\pih{h}}$ and 
check whether $\Post{a}{\pih{h}} \cap Q_f \neq \emptyset$ 
(see Proposition~\ref{prop_universality}). 
If for some $\pih{h}$, we have $\Post{a}{\pih{h}} \cap Q_f = \emptyset$, 
then $\pih{h}$ is a witness of non $a$-universality of $\AutA$, 
since $a \lin{h}\overline{a}$ is not accepted by $\AutA$. 
More precisely, we have the next definition and proposition.

\begin{definition}
\label{def-X_a}
Let $\AutA = (Q, \Sigma, Q_f, \Delta)$ be a hedge automaton, and let $a \in \Sigma$ be a letter. We define
$$\X_a = \{ \pi \in \PointFixeEtoile \mid \Post{a}{\pi} \cap Q_f = \emptyset \}.$$
\end{definition}

\begin{proposition} \label{prop:Xa}
\label{prop-X_a}
$\AutA$ is $a$-universal iff $\X_a = \emptyset$. Moreover, if $\X_a$ is not empty, for all $\pi \in \X_a$, let $h \in \HdgSigma$ be such that $\pi = \pih{h}$. Then $a \lin{h}\overline{a} \in \lin{\TSigma \setminus \La{\AutA}}$.
\end{proposition}

Let us now proceed with the general case, that is, the definition of $\X_{w_ia_i}$ with $i > 1$. 
For all proper prefixes $w_ja_j$ of $w_ia_i$, we can suppose that $\AutA$ is not 
$w_ja_j$-universal, otherwise $\AutA$ would be trivially $w_ia_i$-universal. We define 
$\X_{w_ia_i}$ and then, explain how to check $w_ia_i$-universality knowing $\X_{w_ia_i}$.

\paragraph{General case}

Let $wa$ with $a \in \Sigma$. We first define $\X_{wa}$. 
Let $w = w'a'\lin{h'}$ with $a' \in \Sigma$ and $h' \in \HdgSigma$. 
We suppose that $\AutA$ is not $w'a'$-universal, 
and that $\X_{w'a'} \neq \emptyset$. Moreover  $\X_{w'a'}$ contains a witness 
of a word $v$ such that the tree $t$ with $\lin{t} = w'a'v$ is not accepted 
by $\AutA$. 

\begin{figure}[h!]
\centering
\subfloat[$w'a'\lin{h'}a$]
{
\begin{tikzpicture}
\def\step{1.2}
\def\msize{6mm}

\node at (0,0) {$a'$};
\node[circle, draw, minimum size=\msize] (ap) at (0,0) {};
\node at (-1,-\step) {$h'$};
\node[trapezium, draw, minimum width=1.3cm, minimum height=0.9cm, trapezium stretches=true] (hp) at (-1,-\step-0.2) {};
\node at (1,-\step) {$a$};
\node[circle, draw, minimum size=\msize] (a) at (1,-\step) {};

\node at (0.3,1.2) {$w'$};

\draw[decorate,decoration=snake] (-0.4,1.5) -- (0,0.6);
\draw[dotted] (0,0.6) -- (ap);
\draw (ap.south west) -- (hp.top left corner);
\draw (ap.south west) -- (hp.top right corner);
\draw (ap.south east) -- (a);
\end{tikzpicture}
}
\hspace{10px}
\subfloat[hedge $g$ with $\lin{g} = \lin{h'}a\lin{h_1}\overline{a}\lin{h_2}$]
{
\begin{tikzpicture}
\def\step{1.2}
\def\msize{6mm}

\node at (0,0) {$a'$};
\node[circle, draw, minimum size=\msize] (ap) at (0,0) {};
\node at (-1,-\step) {$h'$};
\node[trapezium, draw, minimum width=1.3cm, minimum height=0.9cm, trapezium stretches=true] (hp) at (-1,-\step-0.2) {};
\node at (0.5,-\step) {$a$};
\node[circle, draw, minimum size=\msize] (a) at (0.5,-\step) {};

\node at (0.5,-2*\step) {$h_1$};
\node[trapezium, draw, minimum width=1.3cm, minimum height=0.9cm, trapezium stretches=true] (h1) at (0.5,-2*\step-0.2) {};

\node at (2,-\step) {$h_2$};
\node[trapezium, draw, minimum width=1.3cm, minimum height=0.9cm, trapezium stretches=true] (h2) at (2,-\step-0.2) {};

\node[trapezium, trapezium angle=75,draw, dotted, minimum width=6cm, minimum height=2.4cm, trapezium stretches=true] (trap) at (0.5,-1.7*\step) {};

\node at (3.3,-1.5*\step) {$g$};
\node at (0.3,1.2) {$w'$};

\draw[decorate,decoration=snake] (-0.4,1.5) -- (0,0.6);
\draw[dotted] (0,0.6) -- (ap);
\draw (ap.south west) -- (hp.top left corner);
\draw (ap.south west) -- (hp.top right corner);
\draw (ap) -- (a);
\draw (ap.south east) -- (h2.top left corner);
\draw (ap.south east) -- (h2.top right corner);
\draw (a) -- (h1.top left corner);
\draw (a) -- (h1.top right corner);
\end{tikzpicture}
}
\caption{Current reading according to the prefix $wa$}
\label{current_reading}
\end{figure}
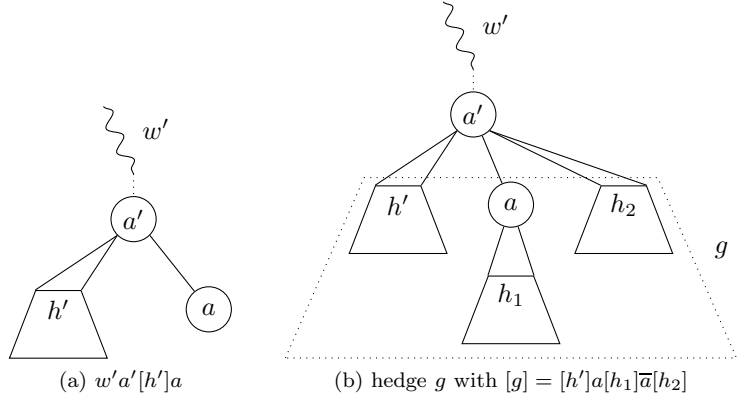

Let us define the set $\X_{wa}$ from the set $\X_{w'a'}$. 
In \Figure~\ref{current_reading}~(a), we indicate the current reading of 
a tree according to $wa$:  an internal node labeled by $a'$ with a sequence 
of subtrees equal to $h'$ followed by a child labeled by $a$. 
With this figure, we notice that $\AutA$ is not $wa$-universal iff there 
exists $h_1, h_2 \in  \HdgSigma$ such that for the hedge $g$ with 
$\lin{g} = \lin{h'}a\lin{h_1}\overline{a}\lin{h_2}$, we have $\pih{g} 
\in \X_{w'a'}$ (see \Figure~\ref{current_reading}~(b)). This observation 
leads to the next definition of $X_{wa}$.

\begin{definition}
\label{def-X_wa}
Let $wa \in \PPref{\TSigma}$ with $w = w'a'\lin{h'}$, $a, a' \in \Sigma$ and $h' \in \HdgSigma$. Let $\AutA = (Q, \Sigma, Q_f, \Delta)$ be a hedge automaton. We define
$$\X_{wa} = \{ \pi \in \PointFixeEtoile \mid \pih{h'} \Post{a}{\pi} \in \Pref{\X_{w'a'}} \}.$$
\end{definition}

As for the basic case (see Proposition~\ref{prop:Xa}), we have the next proposition.

\begin{proposition}
\label{prop:X_ua set of witness}
$\AutA$ is $wa$-universal iff $\X_{wa} = \emptyset$. Moreover, if $\X_{wa}$ is not empty, then $X_{wa} = \{ \pi_h \in \PointFixeEtoile \mid \exists v:
wa\lin{h}\overline{a}v \in \lin{\TSigma \setminus  \La{\AutA}}\}$.
\end{proposition}

\begin{proof}
We proceed by induction on $w$ to prove that $X_{wa} = \{ \pi_h \in \PointFixeEtoile \mid \exists v:
wa\lin{h}\overline{a}v \in \lin{\TSigma  \setminus  \La{\AutA}}\}$.
The basic case, $w = \epsilon$, directly follows from Proposition~\ref{prop:Xa}.

Let $w = w'a'\lin{h'}$ with $a' \in \Sigma$
and $h' \in \HdgSigma$. Suppose that the property holds for $\X_{w'a'}$,
i.e. $X_{w'a'} = \{ \pi_{h'} \in \PointFixeEtoile \mid \exists  v'  : 
 w'a'\lin{h'}v' \in \lin{\TSigma \setminus \La{\AutA}} \}$.

\noindent
($\subseteq$) Let $\pi_h \in \X_{wa}$. By definition, $\exists h',
h'' \in \HdgSigma : \pi_{h'} \Post{a}{\pi_h}\pi_{h''} \in \X_{w'a'}$.
Then, by induction hypothesis, $\exists v'  : w'a'\lin{h'}
a\lin{h}\overline{a}\lin{h''}\overline{a}'v' \in \lin{\TSigma \setminus \La{\AutA}}$.
Let $v = \lin{h''}\overline{a}'v'$, then $wa\lin{h}\overline{a}v \in \lin{\TSigma \setminus \La{\AutA}}$.

\noindent
($\supseteq$) Let $\pi_h \in \PointFixeEtoile$ such that 
$\exists v  : wa\lin{h}\overline{a}v \in \lin{\TSigma \setminus \La{\AutA}}$.
So there exists a word~$v'$ and hedges $h',h''$ such that $w'a'\lin{h'}a\lin{h}\overline{a} 
\lin{h''}\overline{a}'v' \in \lin{\TSigma \setminus \La{\AutA}}$.
By induction hypothesis, $\pi_g \in \X_{w'a'}$ with $\lin{g} = \lin{h'}a\lin{h}\overline{a}\lin{h''}$, and thus
$\pi_h \in \X_{wa}$.
\end{proof}

In this section, given a tree $t_0$ and the current read prefix $u$ of $\lin{t_0}$, we have shown how to test incrementally for $u$-universality as follows. Suppose that $u = a_1\lin{h_1}a_2\lin{h_2} \cdots a_n\lin{h_n}$ with $a_i \in \Sigma, h_i \in \HdgSigma$, for $1 \leq i \leq n$, and let  $w_i = a_1\lin{h_1} \cdots a_{i-1}\lin{h_{i-1}}$, for $1 \leq i \leq n$. We have defined set $\X_a$ and then each set $\X_{w_ia_i}$,  $1 < i \leq n$, from $\X_{w_{i-1}a_{i-1}}$, such that $\AutA$ is $w_ia_i$-universal iff $\X_{w_ia_i}$ is empty. 

It should be noted that it is also possible to test whether $\AutA$ is $w_ia_i[h_i]$-universal thanks to set $\X_{w_ia_i}$. Indeed, by  Proposition~\ref{prop:X_ua set of witness}, $\AutA$ is $w_ia_i[h_i]$-universal iff $\nexists \pi \in \PointFixeEtoile : \pi_{h_i} \pi \in \X_{w_ia_i}$.

\subsubsection{Algorithm for checking $u$-universality}

In this section, we propose an algorithm for $u$-universality checking. As done before for universality in Section~\ref{universality-relations}, we need to represent a macrostate word $\pi$ by the relation $\rel{\pi}$. Definitions~\ref{def-X_a} and~\ref{def-X_wa} are rephrased as follows. Given a set $\Y$ of relations, we define $\Pref{\Y}$ as the set $\{r \in \rel{\PointFixeEtoile} \mid \exists r' \in \rel{\PointFixeEtoile} :  r r' \in   \Y \}$.

\begin{definition} 
\label{def-Y_wa}
Let $\AutA = (Q, \Sigma, Q_f, \Delta)$ be a hedge automaton, and let $wa \in \PPref{\TSigma}$ with $a \in \Sigma$.
\begin{enumerate}
\item If $w = \epsilon$, we define $\Y_a =\{ r \in \rel{\PointFixeEtoile} \mid \Post{a}{r} \cap Q_f = \emptyset \}$.
\item If $w \neq \epsilon$, given $w = w'a'\lin{h'}$ with $a' \in \Sigma$ and $h' \in \HdgSigma$, we define $\Y_{wa} = \{ r \in \rel{\PointFixeEtoile} \mid \rel{\pih{h'}} \rel{\Post{a}{r}} \in \Pref{\Y_{w'a'}} \}$.
\end{enumerate}
\end{definition}

\begin{lemma}
\label{Y_wa and X_wa}
$\Y_{wa} = \rel{\X_{wa}}$.
\end{lemma}

\begin{proof}
The proof is done by induction on $w$. 

The basic case, $\Y_a = \rel{X_a}$, follows from Lemma \ref{post_a_rel_pi}.

Let $w = w'a'\lin{h'}$ with $a' \in \Sigma$ and $h' \in \HdgSigma$.
Suppose that $\Y_{w'a'} = \rel{X_{w'a'}}$ holds. 
Notice that for $\pi \in \X_{wa}$ and $\pi' \in \PointFixeEtoile$, if $\rel{\pi} = \rel{\pi'}$, then $\pi' \in \X_{wa}$ (see Lemma \ref{post_a_rel_pi}). We have for $r = \rel{\pi} \in \rel{\PointFixeEtoile}$:
\begin{eqnarray*}
r \in  \Y_{wa} &\Leftrightarrow& 
\exists r' \in \rel{\PointFixeEtoile}: 
\rel{\pi_{h'}} \rel{\Post{a}{r}} r' \in \Y_{w'a'} \\
&\Leftrightarrow&
\exists \pi' \in \PointFixeEtoile: 
\rel{\pi_{h'}\Post{a}{\pi}\pi'} \in \rel{X_{w'a'}} \\
&\Leftrightarrow& \pi \in \rel{X_{wa}}
\end{eqnarray*}
It follows that $\Y_{wa} = \rel{X_{wa}}$.
\end{proof}

The next proposition is the equivalent of Propositions~\ref{prop-X_a} and~\ref{prop:X_ua set of witness},
as a consequence of Lemma  \ref{Y_wa and X_wa}.

\begin{proposition}
\label{prop-Y_u}
$\AutA$ is $wa$-universal iff $\Y_{wa}$ is empty.
\end{proposition}

By definition of $\Y_{wa}$, it follows that $\AutA$ is $wa\lin{h}$-universal, with $h \in \HdgSigma$, iff $\nexists r \in \rel{ \PointFixeEtoile} : \rel{\pi_{h}} r \in \Y_{wa}$.

Let us now describe an algorithm to test whether a hedge automaton $\AutA$ is $u$-universal. We recall that $t_0$ is a given tree and $u$ its current read prefix. This algorithm is incremental and thus has already checked that $\AutA$ is not $wa$-universal for all non-empty proper prefixes $wa$ of $u$ thanks to Proposition~\ref{prop-Y_u}.

More precisely, let $u = a_1\lin{h_1}a_2\lin{h_2} \cdots a_n\lin{h_n}$ and $w_i = a_1\lin{h_1} \cdots a_{i-1}\lin{h_{i-1}}$, for $1 \leq i \leq n$, and suppose that for all $i$ the sets $\Y_{w_ia_i}$ have been computed and seen to be non empty. A stack is used to store all triples $(\Y_{w_ia_i}, \rel{h_{i}}, a_i)$, $1 \leq i \leq n$, with the triple $(\Y_{w_na_n}, \rel{h_{n}}, a_n)$ at the top of the stack. The stack has a depth equal to the length of $\open{u}$.

In Algorithm~\ref{algo:hedgeu-univ}, four functions are called according 
to the letter that is currently read in $t_0$ knowing that 
$u$ is the last read prefix of $t_0$. 
If it is the first letter~$a$ (resp. last letter $\overline{a}$) 
of $\lin{t_0}$, then Function \textproc{OpenRoot}($a$) 
(resp. \textproc{CloseRoot}($a$)) is called. 
Otherwise either Function \textproc{NextOpenTag}($a$) or 
\textproc{NextClosedTag}($a$) is called according to whether $a$ 
or $\overline{a}$ is the next read letter.

Function \textproc{OpenRoot}($a$) computes the set $\Y_a$ as defined in Definition~\ref{def-Y_wa}. If $\Y_a$ is empty, then $\AutA$ is declared $a$-universal. Otherwise, the stack is initialized with the triple $(\Y_a,\id,a)$ 

Function \textproc{CloseRoot}($a$) pops the stack to get its unique triple $(\Y_a,r,a)$ (since $\overline{a}$ is the last letter of $\lin{t_0}$). It checks whether $t_0 = u\overline{a}$ is accepted by the automaton with the emptiness test of $\Post{a}{r} \cap Q_f$.

If $u\neq \epsilon$ and the letter read after $u$ is $a$ with $a \in \Sigma$, 
then Function \textproc{NextOpenTag}($a$) reads the triple $(\Y',r',a')$ 
at the top of the stack and computes the $\Y_{ua}$ from the set $\Y'$ 
(as in Definition~\ref{def-Y_wa}). 
If $\Y_{ua}$ is empty, then $\AutA$ is declared $ua$-universal. 
Otherwise, the triple $(\Y_{ua},\id,a)$ is pushed on the stack.  
If the letter read after $u$ is $\overline{a}$ with 
$\overline{a} \in \overline{\Sigma}$, and $u\overline{a} \neq t_0$, 
then Function \textproc{NextClosedTag}($a$) pops once the stack 
to get the triple $(\Y,r,a)$ (notice that $\overline{a}$ is 
the closing tag of $a$ in this triple).
It then modifies the triple $(\Y',r',a')$ at the top of stack, 
by replacing $r'$ by $r'' = r' \circ \rel{\Post{a}{r}}$ 
(see \Figure~\ref{current_reading}~(b)). If there does not exist 
$s \in \rel{ \PointFixeEtoile}$ such that $r'' s \in \Y'$, 
then $\AutA$ is declared to be $u\overline{a}$-universal.

These four functions return True as soon as they can declare that 
$\AutA$ is $u$-universal for the current read prefix $u$ of $\lin{t_0}$. 

\begin{algorithm}
\newcommand{\pile}{{\it Stack}}
\newcommand{\push}{\textproc{Push}}
\newcommand{\pop}{\textproc{Pop}}
\newcommand{\Top}{\textproc{Top}}
\begin{algorithmic}
\Function{OpenRoot}{$a$}
    \State $\Y \gets \emptyset$
    \For {$r \in \rel{\PointFixeEtoile}$}
            \If {$\Post{a}{r} \cap Q_f = \emptyset$}
                   \State $\Y \gets \Y \cup \{r\}$  
            \EndIf 
    \EndFor
    \If {$\Y = \emptyset$}
	 \State \Return True \quad  // $u$-universal with $u$ the current read prefix
    \Else
	\State $\pile \gets \emptyset$
	\State $\push(\pile,(\Y,\id,a))$
    \EndIf 
\EndFunction

\medskip
\Function{CloseRoot}{$a$}
     \State $(\Y,r,a) \gets \pop(\pile)$
     \If {$\Post{a}{r} \cap Q_f = \emptyset$}
          \State \Return False \quad // $t_0$ is not accepted 
     \Else
          \State \Return True \quad // $t_0$ is accepted
     \EndIf
\EndFunction

\medskip
\Function{NextOpenTag}{$a$}
     \State $(\Y',r',a') \gets \Top(\pile)$ 
      \State $\Y \gets \emptyset$
    \For {$r \in \rel{\PointFixeEtoile}$}
            \If {$r' \circ \rel{\Post{a}{r}} \in \Pref{\Y'}$}
                   \State $\Y \gets \Y \cup \{r\}$  
            \EndIf 
    \EndFor
    \If {$\Y = \emptyset$}
	 \State \Return True \quad  // $u$-universal with $u$ the current read prefix
	     \Else
	\State $\push(\pile,(\Y,\id,a))$
    \EndIf 
\EndFunction

\medskip
\Function{NextClosedTag}{$a$}
     \State $(\Y,r,a) \gets \pop(\pile)$
     \State $(\Y',r',a') \gets \pop(\pile)$
     \State $r' \gets r' \circ \rel{\Post{a}{r}}$
     \If {$\nexists s \in \rel{ \PointFixeEtoile}: r' \circ s \in \Y' $}
     	\State \Return True \quad  // $u$-universal with $u$ the current read prefix
     \EndIf
     \State $\push(\pile,(\Y',r',a'))$
\EndFunction
\end{algorithmic}
\caption{Functions used for checking $u$-universality incrementally}
\label{algo:hedgeu-univ}
\end{algorithm}

\subsubsection{Antichain-based optimization}

In this section we explain how to use the concept of antichain to
avoid some computations when checking for $u$-universality. 
In particular we show that it is sufficient to 
only compute the $\subseteq$-maximal elements of set $\Y_{wa}$
as defined in Definition~\ref{def-Y_wa}.

\begin{lemma}
Let $wa \in \PPref{\TSigma}$ with $a \in \Sigma$, 
$\Y_{wa}$ is a $\subseteq$-downward closed set.
\end{lemma}

\begin{proof}
We proceed by induction on $w$. Notice that for $r, r' \in \rel{\PointFixeEtoile}$ and $a \in \Sigma$, if $r' \subseteq r$, then $\Post{a}{r'} \subseteq \Post{a}{r}$.

Consider the basic case where $w = \epsilon$. By definition $\Y_{a} = \{r \in  \rel{\PointFixeEtoile} \mid \Post{a}{r} \cap Q_f = \emptyset\}$. By the previous remark, $\Y_{a}$ is a $\subseteq$-downward closed set.

Let $w = w'a'\lin{h'}$, with $a' \in \Sigma$ and $h' \in \HdgSigma$.
Let $r \in \Y_{wa}$ and $r' \in \rel{\PointFixeEtoile}$ such that
$r' \subseteq r$. Let us show that $r' \in \Y_{wa}$.
As $r \in \Y_{wa}$, $\exists r'' \in \rel{\PointFixeEtoile} :
\rel{\pi_{h'}} \rel{\Post{a}{r}} r'' \in \Y_{w'a'}$.
As  $\Post{a}{r'} \subseteq \Post{a}{r}$ and $\Y_{w'a'}$ is $\subseteq$-downward closed,
it follows that $\rel{\pi_{h'}} \rel{\Post{a}{r'}} r'' \in \Y_{w'a'}$ and then
$r' \in \Y_{wa}$.
\end{proof}

As $Y_{wa}$ is $\subseteq$-downward closed, it can be 
described by the antichain  $\lceil \Y_{wa} \rceil$ of its maximal elements.
Let $w = w'a'\lin{h'}$ with $a \in \Sigma$ 
and $h' \in \HdgSigma$, 
the next lemma shows that it is possible to compute $Y_{wa}$ from
$\lceil \Y_{w'a'} \rceil$ without knowing the whole set $\Y_{w'a'}$.

\begin{lemma}
For $r \in  \rel{\PointFixeEtoile}$,  $r \in \Y_{wa}$ iff there exist $r' \in 
\lfloor \rel{\PointFixeEtoile} \rfloor$ and
$s \in \lceil \Y_{w'a'} \rceil$ such that  $\rel{\pi_{h'}} \rel{\Post{a}{r}} r' \subseteq s$.
\end{lemma}

\begin{proof}
\begin{eqnarray*}
r \in \Y_{wa} &\iff&  
\exists r' \in  \rel{\PointFixeEtoile}  : \rel{\pih{h'}} \rel{\Post{a}{r}}r' \in \Y_{w'a'} \quad (\text{Def.~\ref{def-Y_wa}}) \\
&\iff&\exists r' \in \rel{\PointFixeEtoile} , \exists s \in \lceil \Y_{w'a'} \rceil : \rel{\pi_{h'}} \rel{\Post{a}{r}} r' \subseteq s \\
&\iff&\exists r' \in \lfloor \rel{\PointFixeEtoile} \rfloor, \exists s \in \lceil \Y_{w'a'} \rceil : \rel{\pi_{h'}} \rel{\Post{a}{r}} r' \subseteq s
\end{eqnarray*}
\end{proof}

Based on the previous lemma, 
Algorithm \ref{algo:YwaFromYw'a'}  is an optimized version of Function \textproc{NextOpenTag}($u,a$)
which computes $Y = \lceil \Y_{wa} \rceil$ from $Y' = \lceil \Y_{w'a'} \rceil$ without computing the entire set $\Y_{wa}$. The idea is to have a set, called  \emph{Candidates},
containing all elements that could be potentially in $\Y$. 
Initially, it is the set  $\rel{\PointFixeEtoile}$. Otherwise, suppose that 
$\Y$ has been partially computed, then \emph{Candidates} is the set 
$\rel{\PointFixeEtoile} \setminus \{r' \mid \exists r \in \Y : r' \subseteq r \}$.
Function \textproc{MaximalElement}(\emph{Candidates}) returns a maximal element of the set
\emph{Candidates}.

\begin{algorithm}
\newcommand{\pile}{{\it Stack}}
\newcommand{\push}{\textproc{Push}}
\newcommand{\pop}{\textproc{Pop}}
\newcommand{\Top}{\textproc{Top}}
\newcommand{\Candidates}{{\it Candidates}}
\begin{algorithmic}
\Function{OptNextOpenTag}{$u, a$}
     \State $(\Y',r',a') \gets \Top(\pile)$ 
      \State $\Y \gets \emptyset$
      \State $\Candidates \gets \rel{\PointFixeEtoile}$
      \While{$\Candidates \not = \emptyset$}
          \State $r \gets \textproc{MaximalElement}(\Candidates)$
          \If {$\exists r'' \in \lfloor \rel{\PointFixeEtoile} \rfloor,
              \exists s \in \Y' :  r' \circ \rel{\Post{a}{r}} \circ r'' \subseteq s$}
              \State $\Y \gets \Y \cup \{ r \}$
              \State $\Candidates \gets 
              \{ r' \in \Candidates \mid r'\not  \subseteq r \}$ 
         \Else
              \State $\Candidates \gets \Candidates \backslash \{ r \}$
         \EndIf
    \EndWhile
    \If {$\Y = \emptyset$}
	 \State \Return $\AutA$ is $ua$-universal 
    \Else
	\State $\push(\pile,(\Y,\id,a))$
    \EndIf 
\EndFunction
\end{algorithmic}
\caption{Optimized Function  \textproc{NextOpenTag}}
\label{algo:YwaFromYw'a'}
\end{algorithm}

\section{Safe configurations approach} \label{sec:safe}

We present an algorithm for testing $u$-universality of 
a non-deterministic visibly pushdown automaton $\AutA$.
This algorithm is a generalization of the algorithm for the deterministic
case \cite{GauwinNiehrenTison09b}, 
adding several optimizations to avoid huge computations.
As in Section~\ref{sec:HedgeAutomataAndu-universality}, 
the algorithm is incremental in the sense that the 
linearization $\lin{t_0}$ of a given tree $t_0$ is read letter by letter, 
and while $\AutA$ is not $u$-universal for the current read prefix 
$u$ of $\lin{t_0}$, the next letter of $\lin{t_0}$ is read.

\subsection{Safe configurations}

In the deterministic case \cite{GauwinNiehrenTison09b}, the algorithm
relies on the incremental computation of the set of safe states.
In the non-deterministic case, safe states are not enough to decide 
$u$-universa\-lity.
Indeed In \cite{GauwinNiehrenTison09b}, safe states are computed
according to the unique run of the deterministic automaton on $u$.
In fact, safe configurations $(q,\stack)$ are considered, 
but all these configurations have the same stack $\stack$ here, so only
states $q$ have to be stored.
When the automaton is non-deterministic, we may have several runs on $u$,
and each of them may use a different stack.
All these stacks have to be considered for testing $u$-universality,
so we cannot consider only states.

Therefore, we have to consider \emph{safe configurations},
or more precisely sets of safe configurations as described in the next definition.
We use notions about \VPAs that are defined in Section~\ref{sec:VPA}, 
as well sets of configurations
that are antichains with respect to $\subseteq$, 
or $\subseteq$-upward (resp. $\subseteq$-downward) closed sets 
(see Section~\ref{sec:antichain-univ}).

\begin{definition}
Let $\AutA$ be a VPA and $\ConfSet\subseteq Q\times\Gamma^*$ be a set of configurations.
Let $u\in\PPref\TSigma$ be a prefix.
\begin{itemize}
\item $\ConfSet$ \emph{is safe for} $u$ if 
for every $v$ such that $uv\in\lin\TSigma$,
there exist $(q,\stack)\in\ConfSet$ and $p\in Q_f$ such that
$(q,\stack)\xrightarrow{v}(p,\epsilon)$ in $\AutA$. 
\item $\ConfSet$ \emph{is leaf-safe for $u$} if
for every $v=\overline av'$ with $\overline a\in\overline\Sigma$ 
such that $uv\in\lin\TSigma$,
there exist $(q,\stack)\in\ConfSet$ and $p\in Q_f$ such that
$(q,\stack)\xrightarrow{v}(p,\epsilon)$ in $\AutA$. 
\end{itemize}
We write 
$\Safe{u}$ for $\{\ConfSet \mid \ConfSet \text{ is safe for } u\}$
and 
$\LSafe{u}$ for $\{\ConfSet \mid \ConfSet $ is leaf-safe
for $ u\}$.
\end{definition}

Intuitively, as stated in Theorem~\ref{thm:prefix-universal-iff-safe} below,
if $\ConfSet$ is the set of configurations reached in $\AutA$
after reading $u$, then $\AutA$ is $u$-universal iff $\ConfSet$ is safe for $u$.
Indeed, for every possible $v$, one can find in $\ConfSet$ at least one
configuration leading to an accepting configuration after reading $v$.
We first note that, from the definitions, if a set of configurations $\ConfSet$
is safe (resp. leaf-safe) for $u$, then a larger set $\ConfSet'$ is also safe (resp. leaf-safe) for $u$. 

\begin{lemma}
\label{lem:safe-upward-closed}
$\Safe{u}$ and $\LSafe{u}$ are $\subseteq$-upward closed sets.
\end{lemma}

Let $\Reach{u}$ denote the set of configurations $(q,\stack)$ such that
$(q_0,\stack_0)\xrightarrow{u} (q,\stack)$ 
for some initial configuration $(q_0,\stack_0)$ of $\AutA$.

\begin{theorem}
\label{thm:prefix-universal-iff-safe}
$\AutA$ is $u$-universal iff $\Reach{u}\in\Safe{u}$.
\end{theorem}

\begin{proof}
$(\Rightarrow)$ Assume that $\AutA$ is $u$-universal.
Consider the set $\ConfSet$ of configurations $(q,\stack)$ of $\AutA$ such that
there exists
$v\in (\Sigma\cup\overline\Sigma)^*$, $q_i\in Q_i$ and $q_f\in Q_f$ 
verifying $uv\in\lin\TSigma$ and
$(q_i,\epsilon) \xrightarrow {u} (q,\stack) \xrightarrow {v} (q_f,\epsilon)$.
We have $\ConfSet\subseteq\Reach{u}$.

Let $v$ be such that $uv\in\lin\TSigma$. As $\AutA$ is $u$-universal,
there exists a configuration $(q,\stack)\in\ConfSet$ such that 
$(q,\stack) \xrightarrow{v} (q_f,\epsilon)$ with $q_f\in Q_f$.
Hence $\ConfSet\in\Safe{u}$.
By Lemma~\ref{lem:safe-upward-closed}, we get $\Reach{u}\in\Safe{u}$.

$(\Leftarrow)$ Assume now that $\Reach{u}\in\Safe{u}$, and let $v$
be such that $uv\in\lin\TSigma$.
As $\Reach{u}\in\Safe{u}$, there exists $(q,\stack)\in\Reach{u}$ and $p\in Q_f$
such that $(q,\stack) \xrightarrow{v} (p,\epsilon)$.
Thus, $uv\in\La\AutA$, and $\AutA$ is $u$-universal.
\end{proof}

\subsection{Incremental definition of safe configurations}

In this section, we detail how set $\Safe{u}$ of safe configurations can be defined from set $\Safe{u'}$ 
with $u'$ a proper prefix of $u$. In this way, while reading the linearization $\lin{t_0}$ of a given tree $t_0$, set $\Safe{u}$ with $u$ prefix of $\lin{t_0}$, can be incrementally defined. In the next section,
we will turn this approach into an algorithm.

\subsubsection{Starting point}

The starting point is to begin with $\Safe{a}$ for which we recall the definition.
$$\Safe{a} = \{ \ConfSet \mid \forall h \in \HdgSigma,
\exists q_f \in Q_f, \exists (q, \stack) \in \ConfSet :
(q, \stack)  \xrightarrow{h\ol{a}}  (q_f, \epsilon) \}.$$

\subsubsection{Reading a letter $\ol a\in\ol\Sigma$}

When reading an $\overline a\in\overline\Sigma$, we can 
retrieve safe configurations from prior sets of safe configurations:
$$
\Safe{u\overline a} = \Safe{u'} 
$$
where $u'$ is the unique prefix of $u$ such that $u = u' a \lin h$.
Indeed as shown by Lemma~\ref{lem:safe-invariant-by-hedge} below, 
we have $\Safe{u' a \lin h \overline a} = \Safe{u'}$.

Hence, from an algorithmic point of view, 
we just have to use a stack to store these safe configurations.
When opening $a$, we put $\Safe{u'}$ on the stack,
and when closing $\overline a$, we pop it.
As $h$ is a hedge, the stack before reading $\ol a$
is exactly the stack after reading $a$.

\begin{lemma}
\label{lem:safe-invariant-by-hedge}
If $h\in\HdgSigma$, then
$\Safe{u \lin h} = \Safe{u}$
and 
$\LSafe{u \lin h} = \LSafe{u}$.
\end{lemma}

\begin{proof}
$(\supseteq)$ 
Assume $\ConfSet\in\Safe{u}$, 
and let $v$ be such that $u\lin h v\in\lin\TSigma$.
As $h$ is a hedge, we have $uv\in\lin\TSigma$.
As $\ConfSet\in\Safe{u}$, there exists $(q,\stack)\in\ConfSet$ such that
$(q,\stack) \xrightarrow{v} (p,\epsilon)$ with $p\in Q_f$.
So $\ConfSet\in\Safe{u\lin h}$.

$(\subseteq)$
Conversely, assume $\ConfSet\in\Safe{u\lin h}$.
Let $v$ be such that $uv\in\lin\TSigma$.
We also have $u\lin h v\in\lin\TSigma$, so there exists 
$(q,\stack)\in\ConfSet$ such that 
$(q,\stack) \xrightarrow{v} (p,\epsilon)$ with $p\in Q_f$.
Thus $\ConfSet\in\Safe{u}$.

The proof is the same for $\LSafe{u \lin h} = \LSafe{u}$,
except that we only consider $v$ of the form $\overline a v'$.
\end{proof}

In the rest of Section~\ref{sec:safe}, we only treat sets $\Safe{ua}$ since the way of computing sets $\Safe{u\overline{a}}$ has been just detailed. The case of sets $\Safe{ua}$ is much more involved.

\subsubsection{Reading a letter $a\in\Sigma$}
%
When reading an $a\in\Sigma$,
two successive steps are performed,
with leaf-safe configurations as intermediate object:
$$
\Safe{u} \qquad \xrightarrow{\text{Step 1}} \qquad \LSafe{ua} 
\qquad \xrightarrow{\text{Step 2}} \qquad \Safe{ua}
$$ 
We now detail Step 1 and Step 2, i.e.
how $\LSafe{ua}$ can be defined from $\Safe{u}$,
and how $\Safe{ua}$ is defined from $\LSafe{ua}$.
Proposition~\ref{prop:lsafe-safe} gives a first idea of these links.
Equivalence~(\ref{eqn:safe-lsafe}) states that a set of configurations 
$\ConfSet$ is leaf-safe for $ua$ iff after performing a $\Post{\ol a}{\ConfSet}$
we get a safe set of configurations for $u$.
Equivalence~(\ref{eqn:lsafe-safe}) states that safe configurations for $ua$
are those from which traversing any hedge leads to a leaf-safe 
set of configurations, i.e. one can safely close the $a$-node.
Proposition~\ref{prop:lsafe-safe} thus relates sets $\Safe{u}$, $\LSafe{ua}$, and $\Safe{ua}$, 
however backwardly. Proposition~\ref{prop:lsafe-safe-pred} hereafter will relates them in the
right direction.

\begin{proposition} \label{prop:Post}
\label{prop:lsafe-safe}
Let $ua\in\PPref\TSigma$ with $a \in \Sigma$.
\begin{eqnarray}
\ConfSet\in\LSafe{ua} \iff \Post{\overline a}{\ConfSet} \in \Safe{u} 
\label{eqn:safe-lsafe}\\
\ConfSet\in\Safe{ua} \iff
\forall h\in\HdgSigma, 
\Post{\lin h}{\ConfSet} \in \LSafe{ua}
\label{eqn:lsafe-safe}
\end{eqnarray}
\end{proposition}

\begin{proof}
$(\ref{eqn:safe-lsafe}, \Rightarrow)$
Let $\ConfSet \in \LSafe{ua}$ and $\ConfSet'=\Post{\overline a}{\ConfSet}$.
Let us show that $\ConfSet'\in\Safe{u}$.
By Lemma~\ref{lem:safe-invariant-by-hedge},
it is sufficient to prove that $\ConfSet'\in\Safe{ua\overline a}$.
Let $v$ such that $ua\overline a v\in\lin\TSigma$.
As $\ConfSet\in\LSafe{ua}$ and 
$\overline a v$ starts with $\overline a\in\overline\Sigma$,
there exists $(q,\stack)\in\ConfSet$ and $(q',\stack')$ such that 
$(q,\stack) \xrightarrow{\overline a} (q',\stack')
\xrightarrow{v} (p,\epsilon)$ for some $p\in Q_f$.
By definition of $\Post{\overline a}{\ConfSet}$ we have 
$(q',\stack')\in\ConfSet'$ and thus $\ConfSet'\in\Safe{ua\overline a}$.

$(\ref{eqn:safe-lsafe}, \Leftarrow)$
For the converse, let $\ConfSet'=\Post{\overline a}{\ConfSet}
\in \Safe{u}=\Safe{ua\overline a}$. Let us show that $\ConfSet\in\LSafe{ua}$.
Let $v$ be such that $uav\in\lin\TSigma$ and $v=\overline b v'$.
We necessarily have $\overline a=\overline b$.
As $\ConfSet'\in\Safe{ua\overline a}$, there exists $(q',\stack')\in\ConfSet'$
such that $(q',\stack') \xrightarrow{v'} (p,\epsilon)$ with $p\in Q_f$.
By definition of $\Post{\overline a}{\ConfSet}$, there also exists
$(q,\stack)\in\ConfSet$ such that 
$(q,\stack) \xrightarrow{\overline a}(q',\stack')$
and thus $(q,\stack) \xrightarrow{v=\overline a v'} (p,\epsilon)$ 
with $p\in Q_f$.

$(\ref{eqn:lsafe-safe}, \Rightarrow)$
Let $\ConfSet\in\Safe{ua}$ and $h\in\HdgSigma$.
Let us show that $\ConfSet'=\Post{\lin h}{\ConfSet}$ is in $\LSafe{ua}$.
Let $v$ such that $uav\in\lin\TSigma$ and $v=\overline b v'$.
We must have $\overline a=\overline b$.
We also have $ua\lin{h}\overline av'\in\lin\TSigma$.
As $\ConfSet\in\Safe{ua}$, there exists $(q,\stack)\in\ConfSet$ such that
$(q,\stack) \xrightarrow{\lin h} (q',\stack') \xrightarrow{v=\overline a v'}
(p,\epsilon)$ with $p\in Q_f$.
By definition of $\Post{\lin h}{\ConfSet}$, $(q',\stack')\in\ConfSet'$.

$(\ref{eqn:lsafe-safe}, \Leftarrow)$
Let us assume that for every hedge $h\in\HdgSigma$, 
$\Post{\lin h}{\ConfSet}\in\LSafe{ua}$.
Let us show that $\ConfSet\in\Safe{ua}$.
Let $v$ be such that $uav\in\lin\TSigma$.
Then we have $uav=ua\lin{h}\overline av'$ for some $h \in \HdgSigma$.
As $\ConfSet'=\Post{\lin{h}}{\ConfSet}\in\LSafe{ua}$,
$\ol a v'$ starts with $\ol a\in\ol\Sigma$
and $ua\ol a v'\in\lin\TSigma$,
there exists $(q',\stack')\in\ConfSet'$ such that 
$(q',\stack') \xrightarrow{\ol a v'} (p,\epsilon)$ for some $p\in Q_f$.
Hence, by definition of $\Post{\lin{h}}{\ConfSet}$, there also exists
$(q,\stack)\in\ConfSet$ such that
$(q,\stack) \xrightarrow{\lin{h}} (q',\stack') 
\xrightarrow{\ol a v'} (p,\epsilon)$ with $p\in Q_f$.
\end{proof}

We propose now the notion of \emph{predecessor} in a way to get Step~1 and Step~2 in the right direction.
\begin{definition}
Let $\ConfSet,\ConfSet'$ be two sets of configurations,
$\ol a\in\ol\Sigma$ and $h\in\HdgSigma$.
\begin{itemize}
\item $\ConfSet$ \emph{is an $\ol a$-predecessor of} $\ConfSet'$
if $\forall (q',\stack')\in\ConfSet',\ 
\exists (q,\stack)\in\ConfSet,\ 
(q,\stack) \xrightarrow{\ol a} (q',\stack')$.
\item $\ConfSet$ \emph{is an $h$-predecessor of} $\ConfSet'$
if $\forall (q',\stack')\in\ConfSet',\ 
\exists (q,\stack)\in\ConfSet,\ 
(q,\stack) \xrightarrow{\lin{h}} (q',\stack')$.
\end{itemize}
Let $\Pred{\ol a}{\ConfSet'} = \{ \ConfSet \mid 
\ConfSet \text{ is an $\ol a$-predecessor of }\ConfSet'\}$ and 
$\Pred{h}{\ConfSet'} = \{ \ConfSet \mid 
\ConfSet \text{ is an $h$-predecessor of }\ConfSet'\}$.
\end{definition}

From their definitions, the sets of predecessors are $\subseteq$-upward closed.

\begin{lemma}
\label{lem:pred-upward-closed} 
$\Pred{\ol a}{\ConfSet'}$ and $\Pred{h}{\ConfSet'}$ are $\subseteq$-upward closed sets.
\end{lemma}

Predecessors closely relate to the $\PostOp$ operator.

\begin{lemma}
\label{lem:post-and-pred} 
$\ConfSet$ is an $\ol a$-predecessor of $\Post{\ol a}{\ConfSet}$.
If $\ConfSet$ is an $\ol a$-predecessor of $\ConfSet'$ then
$\ConfSet' \subseteq \Post{\ol a}{\ConfSet}$.
Both properties also hold for $\Post{\lin{h}}{\ConfSet}$.
\end{lemma}

We can now rephrase Proposition~\ref{prop:lsafe-safe} in terms of predecessors.

\begin{proposition}
\label{prop:lsafe-safe-pred} 
Let $ua\in\PPref\TSigma$.
\begin{eqnarray}
\ConfSet\in\LSafe{ua} \iff \exists \ConfSet'\in\Safe{u},\ 
\ConfSet \text{ is an $\ol a$-predecessor of } \ConfSet'
\label{eqn:safe-lsafe-pred}\\
\ConfSet\in\Safe{ua} \iff
\text{$\forall h\in\HdgSigma$, }
\exists \ConfSet'\in\LSafe{ua},\ 
\ConfSet \text{ is a $h$-predecessor of } \ConfSet'
\label{eqn:lsafe-safe-pred}
\end{eqnarray}
\end{proposition}

\begin{proof}
$(\ref{eqn:safe-lsafe-pred},\Rightarrow)$
Let $\ConfSet\in\LSafe{ua}$.
Then by Proposition~\ref{prop:lsafe-safe},
$\Post{\ol a}{\ConfSet} \in \Safe{u}$.
Moreover, $\ConfSet$ is an $\ol a$-predecessor of $\Post{\ol a}{\ConfSet}$
by Lemma~\ref{lem:post-and-pred}.

$(\ref{eqn:safe-lsafe-pred},\Leftarrow)$
Let $\ConfSet$ be an $\ol a$-predecessor of $\ConfSet'$, with 
$\ConfSet'\in\Safe{u}$.
By Lemma~\ref{lem:post-and-pred}, 
$\ConfSet'\subseteq\Post{\ol a}{\ConfSet}$.
By Lemma~\ref{lem:safe-upward-closed},
we also have $\Post{\ol a}{\ConfSet}\in\Safe{u}$, so
$\ConfSet\in\LSafe{ua}$ by Proposition~\ref{prop:lsafe-safe}.

$(\ref{eqn:lsafe-safe-pred})$ Same proofs, except that $\ol a$ has to be
replaced by $h$, for all $h\in\HdgSigma$.
\end{proof}

Proposition~\ref{prop:lsafe-safe-pred} can be used to perform
Step~1 and Step~2 of our method.
It states that safe sets of configurations are only
among predecessors of prior safe sets of configurations.
However, the number of hedges to consider in 
equivalence~(\ref{eqn:lsafe-safe-pred}) is infinite.
We use relations to overcome this.
Also the size of $\Safe{u}$ may be huge
and not all configurations of $\Safe{u}$ are crucial for checking
$u$-universality. We use antichains to have a representation of 
$\Safe{u}$ and to avoid computations of elements which are not crucial.
These two concepts are  explained in the following 
in a way to get an algorithm for incrementally checking $u$-universality.

\subsection{An algorithm for $u$-universality}

\subsubsection{Antichains}

Let $\subsetatc{\Safe{u}}$ denote the set of 
elements of $\Safe{u}$ which are minimal for $\subseteq$, similarly for $\LSafe{u}$.
These antichains are finite objects.

\begin{proposition} \label{prop:finiteObjects}
\label{prop:safe-lsafe-antichains} 
$\subsetatc{\Safe{u}}$ and $\subsetatc{\LSafe{u}}$ are
finite and only contain finite sets of configurations.
\end{proposition}

\begin{proof}
We begin with the following observation.
Let $v$ be such that $\lin{uv}\in\TSigma$ and 
$(q,\stack) \xrightarrow{v} (p,\epsilon)$ with $p\in Q_f$.
Let $u' = \open{u}$ (recall that $\open{u}$ is the word obtained from $u$ by removing all factors that are linearizations of hedges).
Let $v'$ be the word obtained from $v$ in the same way.
Then $|u'|=|v'|$ and $|u'|=|\stack|$.

Let $\ConfSet \in \Safe{u}$. Then by definition
$$\forall v,\ uv\in\lin\TSigma \implies 
\exists (q,\stack)\in\ConfSet,\ 
(q,\stack) \xrightarrow{v} (p,\epsilon) \text{ with } p\in Q_f.$$
If $\ConfSet$ is minimal with respect to $\subseteq$, then
every $(q,\stack)\in\ConfSet$ is used for at least one $v$ in the
previous definition. Now by the previous observation,
each such $(q,\stack)$ belongs to $Q\times\Gamma^{|u'|}$.
Hence $\ConfSet\subseteq Q\times\Gamma^{|u'|}$,
and thus both $\ConfSet$ and $\subsetatc{\Safe{u}}$ are finite.

The same arguments hold for proving that $\subsetatc{\LSafe{u}}$ is finite
and contains only finite sets of configurations.
\end{proof}

We now try to use these antichains in the starting point, and in Steps~1 and~2 of our approach.

\subsubsection{Step~1 with antichains: 
from $\subsetatc{\Safe{u}}$ to $\subsetatc{\LSafe{ua}}$}

For the two steps, the goal is to adapt Proposition~\ref{prop:lsafe-safe-pred}
so that it uses $\subsetatc{\Safe{.}}$ instead of $\Safe{.}$,
and $\subsetatc{\LSafe{.}}$ instead of $\LSafe{.}$.
We begin with Step 1. Implication $(\Rightarrow)$ of equivalence~(\ref{eqn:safe-lsafe-pred}) 
can be directly adapted.

\begin{proposition}
\label{prop:lsafe-safe-atc} 
Let $ua\in\PPref\TSigma$.
$$
\ConfSet\in\subsetatc{\LSafe{ua}} \implies
\exists \ConfSet'\in\subsetatc{\Safe{u}},\ 
\text{ $\ConfSet$ is an $\ol a$-predecessor of $\ConfSet'$}
$$
\end{proposition}

\begin{proof}
Let $\ConfSet\in\subsetatc{\LSafe{ua}}$ and 
let $\ConfSet'=\Post{\ol a}{\ConfSet}$.
We know from Proposition~\ref{prop:lsafe-safe} that $\ConfSet'\in\Safe{u}$.
Let $\ConfSet_0'\subseteq\ConfSet'$ such that 
$\ConfSet_0'\in\subsetatc{\Safe{u}}$.
From the definition of $\ConfSet'$ we get:
$$
\forall c'\in\ConfSet',\ \exists c\in\ConfSet,\ c\xrightarrow{\ol a} c'
$$
We build $\ConfSet_0$ from these $c\in\ConfSet$ but for $c'\in\ConfSet_0'$:
$$
\ConfSet_0 = \{ c\in\ConfSet \hmid 
\exists c'\in\ConfSet_0',\ c\xrightarrow{\ol a} c'\}
$$
\Figure~\ref{fig:construction of ConfSet_0} illustrates
the construction.
\begin{figure}[h!]
\centering
\begin{tikzpicture}
\draw (0,0) ellipse (1cm and 2cm);
\draw (0,0.1) ellipse (0.6cm and 1cm);
\draw (4,0) ellipse (1cm and 2cm);
\draw (4.05,-0.1) ellipse (0.6cm and 1cm);

\node[circle, fill=black, inner sep= 1pt] (l1) at (0.1,1.4) {};
\node[circle, fill=black, inner sep= 1pt] (l2) at (0.2,0.8) {};
\node[circle, fill=black, inner sep= 1pt] (l3) at (-0.2,0.2) {};
\node[circle, fill=black, inner sep= 1pt] (l4) at (0.2,-0.4) {};
\node[circle, fill=black, inner sep= 1pt] (l5) at (-0.2,-0.6) {};
\node[circle, fill=black, inner sep= 1pt] (l6) at (-0.3,-1.4) {};

\node[circle, fill=black, inner sep= 1pt] (r1) at (4.2,1.7) {};
\node[circle, fill=black, inner sep= 1pt] (r2) at (4.3,1.25) {};
\node[circle, fill=black, inner sep= 1pt] (r3) at (4.1,0.7) {};
\node[circle, fill=black, inner sep= 1pt] (r4) at (4.2,0.25) {};
\node[circle, fill=black, inner sep= 1pt] (r5) at (3.9,0) {};
\node[circle, fill=black, inner sep= 1pt] (r6) at (4.3,-0.6) {};
\node[circle, fill=black, inner sep= 1pt] (r7) at (3.8,-1.25) {};
\node[circle, fill=black, inner sep= 1pt] (r8) at (3.9,-1.7) {};

\draw[->] (l1) -- (r1);
\draw[->] (l2) -- (r2);
\draw[->] (l2) -- (r3);
\draw[->] (l3) -- (r4);
\draw[->] (l3) -- (r5);
\draw[->] (l4) -- (r6);
\draw[->] (l5) -- (r6);
\draw[->] (l6) -- (r7);
\draw[->] (l6) -- (r8);
\draw[->] (l5) -- (r7);

\node at (-1,1.4) {$\ConfSet$};
\node at (5,1.4) {$\ConfSet'$};
\node at (-0.5,-1) {$\ConfSet_0$};
\node at (4.5,-1.2) {$\ConfSet_0'$};
\node at (2,1.75) {$\overline{a}$};
 
\end{tikzpicture}
\caption{Construction of $\ConfSet_0$}
\label{fig:construction of ConfSet_0}
\end{figure}
$\ConfSet_0$ is an $\ol a$-predecessor of $\ConfSet_0'$, so using 
Proposition~\ref{prop:lsafe-safe-pred},
we get $\ConfSet_0\in\LSafe{ua}$.
Furthermore, $\ConfSet_0\subseteq\ConfSet\in\subsetatc{\LSafe{ua}}$,
so $\ConfSet_0=\ConfSet$, and $\ConfSet$ is obtained as an
$\ol a$-predecessor of $\ConfSet_0'\in\subsetatc{\Safe{u}}$.
\end{proof}

Proposition~\ref{prop:lsafe-safe-atc} gives us a way to compute
$\subsetatc{\LSafe{ua}}$ from $\subsetatc{\Safe{u}}$:
it suffices to take all $\ol a$-predecessors of elements of 
$\subsetatc{\Safe{u}}$ and then limit to those predecessors that are $\subseteq$-minimal.
We can even only consider minimal $\ol a$-predecessors of 
$\subsetatc{\Safe{u}}$ in the following sense: $\ConfSet$ is a \emph{minimal $\ol a$-predecessor} of $\ConfSet'$ if for all $\ConfSet''$ $\ol a$-predecessor of $\ConfSet'$,
$\ConfSet''\subseteq\ConfSet \implies \ConfSet''=\ConfSet$.
We finally obtain:
\begin{corollary} \label{cor:Step1Antichains}
$$
\subsetatc{\LSafe{ua}} =
\subsetatc{\left\lbrace\ConfSet \hmid \ConfSet 
\text{ is a minimal $\ol a$-predecessor of } \ConfSet'\in\subsetatc{\Safe{u}}
\right\rbrace}
$$
\end{corollary}

\subsubsection{Step~2 with antichains: 
from $\subsetatc{\LSafe{ua}}$ to $\subsetatc{\Safe{ua}}$}

The second step for computing $\subsetatc{\Safe{ua}}$ from
$\subsetatc{\Safe{u}}$ relies on the introduction of antichains in 
equivalence~(\ref{eqn:lsafe-safe-pred}) of 
Proposition~\ref{prop:lsafe-safe-pred}.
Implication $(\Rightarrow)$ holds with antichains.

\begin{proposition}
\label{prop:safe-lsafe-atc} 
Let $ua\in\PPref\TSigma$.
$$
\ConfSet\in\subsetatc{\Safe{ua}} \implies
\text{$\forall h\in\HdgSigma$, }
\exists \ConfSet'\in\subsetatc{\LSafe{ua}},\ 
\ConfSet \text{ is a $h$-predecessor of } \ConfSet'
$$
\end{proposition}

\begin{proof}
The proof is in the same vein as for Proposition~\ref{prop:lsafe-safe-atc}. 
Let $\ConfSet\in\subsetatc{\Safe{ua}}$, and $h\in\HdgSigma$.
Let $\ConfSet_h'=\Post{\lin{h}}{\ConfSet}$.
By Proposition~\ref{prop:lsafe-safe},
$\ConfSet_h'\in\LSafe{ua}$.
Let $\ConfSet_h''\subseteq\ConfSet_h'$
such that $\ConfSet_h''\in\subsetatc{\LSafe{ua}}$.
We know that $\forall c'\in\ConfSet_h',\ \exists c\in\ConfSet$ such that
${c}\xrightarrow{\lin{h}}{c'}$.
We define $\ConfSet_h = 
\{ c\in\ConfSet \hmid \exists c'\in\ConfSet_h'',\ c\xrightarrow{\lin{h}} c' \}$.
For every $h\in\HdgSigma$, $\ConfSet_h$ is a $h$-predecessor
of $\ConfSet_h''\in\LSafe{ua}$.
Consider $\ConfSet_\cup = \bigcup_{h\in\HdgSigma} \ConfSet_h$, then
$\ConfSet_\cup$ is also a $h$-predecessor of $\ConfSet_h''$.
Using Proposition~\ref{prop:lsafe-safe-pred},
we have $\ConfSet_\cup\in\Safe{ua}$.
As $\ConfSet_\cup\subseteq\ConfSet$ and $\ConfSet\in\subsetatc{\Safe{ua}}$,
we also have that $\ConfSet_\cup=\ConfSet$.
Hence $\ConfSet$ verifies that
$\forall h\in\HdgSigma$, $\exists \ConfSet''\in\subsetatc{\LSafe{ua}}$ 
such that 
$\ConfSet$ is a $h$-predecessor of $\ConfSet''$.
\end{proof}

Note that this proof does not use the fact that $ua$ ends with a symbol
in $\Sigma$, so Proposition~\ref{prop:safe-lsafe-atc} also holds
when replacing $ua$ by $u$.

Similarly to Proposition~\ref{prop:lsafe-safe-atc},
we can restrict $h$-predecessors to consider to only minimal ones:
$\ConfSet$ is a \emph{minimal $h$-predecessor} of $\ConfSet'$
if for all $\ConfSet''$ $h$-predecessor of $\ConfSet'$,
$\ConfSet''\subseteq\ConfSet \implies \ConfSet''=\ConfSet$.
We obtain: 
\begin{corollary} \label{cor:Safe(ua)}
$$
\subsetatc{\Safe{ua}} =
\subsetatc{\left\lbrace\ConfSet \hmid \ConfSet=\bigcup_{h\in\HdgSigma} \ConfSet_h
\text{ with $\ConfSet_h$ a minimal $h$-predecessor of } 
\ConfSet'\in\subsetatc{\LSafe{ua}}\right\rbrace}
$$
\end{corollary}

This definition does not provide an algorithm, as it
still relies on a quantification
over an infinite number of hedges $h\in\HdgSigma$.
In fact, only a finite number of such hedges needs to be considered.
The reason is that a hedge does not change the original stack during
the run of a \VPA, so a hedge can be considered as a function
mapping each state $q$ to the set of states obtained when traversing $h$
from $q$.
Formally, we have the next definition.

\begin{definition}
For every $h\in\Hdg\Sigma$, $\vparel{h}$ is the function from $Q$ to $2^Q$ such that
$q'\in\vparel{h}(q)$ iff $(q,\stack) \xrightarrow{\lin h} (q',\stack)$
for some $\stack\in\Gamma^*$.
\end{definition}
The number of such functions is finite, and bounded by $|Q|\cdot 2^{|Q|}$. 
These functions naturally define an equivalence relation of
finite index over $\HdgSigma$:
$$
h\vpaheq h' \iff \vparel{h} = \vparel{h'}.
$$
Let us note $\finiteH$ for a subset containing one hedge per $\vpaheq$-class.
We have $|\finiteH| \le |Q|\cdot 2^{|Q|}$. The next lemma indicates that the computation
of $h$-predecessors can be limited to $h \in H$.

\begin{lemma}
\label{prop:finitely-many-hedges} 
For every $h\in\HdgSigma$, $\ConfSet$ is a $h$-predecessor of $\ConfSet'$ iff 
there exists $h' \in \finiteH, h\vpaheq h'$, 
such that  $\ConfSet$ is a $h'$-predecessor of $\ConfSet'$.
\end{lemma}

\begin{proof}
Let us recall the definition of $h$-predecessor:
$\ConfSet$ is a $h$-predecessor of $\ConfSet'$ if
$\forall (q',\stack)\in\ConfSet',\ 
\exists (q,\stack)\in\ConfSet,\ 
(q,\stack) \xrightarrow{h} (q',\stack)$.
Hence if $h\vpaheq h'$ then 
$\ConfSet$ is a $h$-predecessor of $\ConfSet'$ iff
$\ConfSet$ is a $h'$-predecessor of $\ConfSet'$.
\end{proof}

We propose an algorithm for computing such a set $\finiteH$ from a VPA $\AutA$.
Algorithm~\ref{algo:vparel} is based on the definition of hedges,
adapted to relations:
\begin{itemize}
\item $\epsilon$ is the empty hedge, and $\vparel{\epsilon}(q)=\{q\}$
for every $q\in Q$. We write this function $\vpaid_Q$.
\item if $h_1,h_2$ are two hedges, then $h_1h_2$ is a hedge,
and $\vparel{h_1h_2}=\vparel{h_2}\circ\vparel{h_1}$.
\item if $h$ is a hedge and $a\in\Sigma$, then
$ah\ol a$ is a hedge, and $\vparel{ah\ol a}(q)$ is the set of states $q'$
such that there exists $\gamma\in\Gamma$ verifying:
$$
(q,\epsilon) \xrightarrow{a} (p,\gamma)
\quad \text{ and } \quad
(p',\gamma) \xrightarrow{\ol a} (q',\epsilon)
\quad \text{ with } p'\in\vparel{h}(p).
$$
\end{itemize}
Algorithm ~\ref{algo:vparel} uses the variables \emph{ToProcess} and \emph{Functions} with the following meaning. \emph{Functions} contains initially the identity relation $\vpaid_Q$; at the end of the computation, it contains all functions $\vparel{h}$, for $h\in\HdgSigma$. \emph{ToProcess} contains all the newly constructed relations, and these relations are used to create other new relations as described in the previous definition by induction.   

\begin{algorithm}
\newcommand{\allfuncs}{{\it Functions}}
\newcommand{\toprocess}{{\it ToProcess}}
\newcommand{\newfuncs}{{\it NewFunctions}}
\newcommand{\compose}{{\it compose}}
\newcommand{\addroot}{{\it add\_root}}
\newcommand{\func}{{\it fct}}
\newcommand{\compositions}{{\it compositions}}

\begin{algorithmic}
\Function{HedgeFunctions}{$\AutA$}
    \State $\allfuncs \gets \{ \vpaid_Q \}$
    \State $\toprocess \gets \{ \vpaid_Q \}$
    \While{$\toprocess \not= \emptyset$}
        \State $\func \gets \textproc{Pop}(\toprocess)$
        \State $\newfuncs \gets \emptyset$
        \For {$f\in\allfuncs$}
            \State $\newfuncs \gets 
            \newfuncs \cup \{ f \circ \func, \func \circ f \}$
        \EndFor
        \For {$a\in\Sigma$}
            \State $f \gets f_\emptyset$ 
            \qquad // $f_\emptyset$ maps every $q\in Q$ to $\emptyset$
            \For {$\vparule{q}{a}{\gamma}{p}\in\Delta$ and 
              $\vparule{p'}{\ol a}{\gamma}{q'}\in\Delta$ with $p'\in\func(p)$}
                \State $f(q) \gets f(q) \cup \{q'\}$
            \EndFor
            \State $\newfuncs \gets \newfuncs \cup \{f\}$
        \EndFor
        \State $\toprocess \gets \toprocess \cup (\newfuncs \setminus \allfuncs)$
        \State $\allfuncs \gets \allfuncs \cup \newfuncs$
    \EndWhile
    \State \Return \allfuncs
\EndFunction
\end{algorithmic}
\caption{Computing all functions $\vparel{h}$, for $h\in\HdgSigma$.
\label{algo:vparel}}
\end{algorithm}

\begin{proposition}
Algorithm~\ref{algo:vparel} computes the set $\{ \vparel{h} \mid h\in\HdgSigma\}$.
\end{proposition}

\begin{proof}
Let \emph{Functions} be the set computed by Algorithm~\ref{algo:vparel}. 
Clearly, $\text{\emph{Functions}} \subseteq  \{ \vparel{h} \mid h\in\HdgSigma\}$. 
Assume for contradiction that there exists $r = \vparel{h} $ 
with $h \in \HdgSigma$ such that $r \not\in \text{\emph{Functions}}$. 
Clearly, $r \neq \vpaid_Q$, and we can suppose wlog that 
either $r = r'_2 \circ r'_1$ with 
$r'_1, r'_2 \in \text{\emph{Functions}} \setminus \{\vpaid_Q\}$, 
or there exists $r' \in \text{\emph{Functions}}$ such that 
for all $q$, $r(q)$ is the set of $q'$ 
with $\vparule{q}{a}{\gamma}{p}\in\Delta$, 
$\vparule{p'}{\ol a}{\gamma}{q'}\in\Delta$ and $p'\in r'(p)$. 
Consider the first case. 
When they have been constructed by Algorithm~\ref{algo:vparel}, 
both $r'_1$ and $r'_2$ have been added to \emph{ToProcess} and 
to \emph{Functions}. After the last element (among $r'_1$ and $r'_2$) 
is popped from \emph{ToProcess}, then $r = r'_2 \circ r'_1$ is built 
during the loop on $f \in \text{\emph{Functions}}$, 
which leads to a contradiction. 
We also have a contradiction in the second case by considering 
the loop on $a \in \Sigma$.    
\end{proof}

Consequently we can rephrase our definition of $\subsetatc{\Safe{ua}}$
from $\subsetatc{\LSafe{ua}}$ given in Corollary~\ref{cor:Safe(ua)} 
by restricting the quantification on $h$ to the finite set $H$. 
Therefore we obtain a finite procedure for computing $\subsetatc{\Safe{ua}}$
from $\subsetatc{\LSafe{ua}}$: 
\begin{proposition}
\label{prop:safe-lsafe-atc-H} 
$$
\subsetatc{\Safe{ua}} =
\subsetatc{\left\lbrace\ConfSet \hmid \ConfSet=\bigcup_{h\in\finiteH} \ConfSet_h
\text{ with $\ConfSet_h$ a minimal $h$-predecessor of } 
\ConfSet'\in\subsetatc{\LSafe{ua}}\right\rbrace}
$$
\end{proposition}

\subsubsection{Starting point with antichains}

It remains to explain how to compute $\Safe{a}$. Clearly, by definition of $H$, we can 
compute $\lfloor \Safe{a} \rfloor$ as follows:

\begin{proposition}
$$
\subsetatc{\Safe{a}} = 
\subsetatc{ \left\lbrace \ConfSet \mid \forall h \in \finiteH,
\exists q_f \in Q_f, \exists (q, \stack) \in \ConfSet :
(q, \stack)  \xrightarrow{h\ol{a}}  (q_f, \epsilon) \right\rbrace}.$$
\end{proposition}

\subsection{Algorithmic improvements}

The previous section resulted in a first algorithm
to incrementally compute sets of safe configurations.
This algorithm can be improved by limiting 
hedges to consider, and optimizing operators and predecessors to be computed. 
The goal here is to avoid the complexity of the on-the-fly
determinization procedure.

\subsubsection{Minimal hedges}

A first improvement is obtained by further restricting hedges to consider.
Indeed it suffices to consider \emph{minimal hedges}
wrt their function $\vparel{h}$.
Formally, let us write $h\le h'$ whenever
$\vparel{h}(q) \subseteq \vparel{h'}(q)$ for every $q\in Q$.
We denote by $\subsetatc{\finiteH}$ the $\le$-minimal elements of $\finiteH$.
Notice that Algorithm~\ref{algo:vparel} that computes the set $\{ \vparel{h} \mid h \in \finiteH\}$ can be easily adapted to compute the set of its minimal elements, such that \emph{NewFunctions} and \emph{ToProcess} 
are restricted to antichains of minimal elements.

From the definition of $h$-predecessor, for
every $\ConfSet,\ConfSet'\in Q\times\Gamma^*$ we have:
\begin{equation} \label{eqn:min-hedges-and-pred}
\ConfSet \text{ $h$-predecessor of }\ConfSet' \text{ and } h\le h' \implies 
\ConfSet \text{ $h'$-predecessor of }\ConfSet'
\end{equation}
This property can be used to replace $h\in\finiteH$ in 
Proposition~\ref{prop:safe-lsafe-atc-H} by $h\in\subsetatc{\finiteH}$.

\begin{proposition} \label{prop:Step2Improved}
\label{prop:safe-lsafe-atc-Hatc} 
$$
\subsetatc{\Safe{ua}} =
\subsetatc{\left\lbrace\ConfSet \hmid \ConfSet=\bigcup_{h\in\subsetatc{\finiteH}} 
\ConfSet_h
\text{ with $\ConfSet_h$ a minimal $h$-predecessor of } 
\ConfSet'\in\subsetatc{\LSafe{ua}}\right\rbrace}
$$
\end{proposition}

\begin{proof}
Let $S$ denote the set 
$$\left\lbrace\ConfSet \hmid \ConfSet=\bigcup_{h\in\finiteH} 
\ConfSet_h
\text{ with $\ConfSet_h$ a minimal $h$-predecessor of } 
\ConfSet'\in\subsetatc{\LSafe{ua}}\right\rbrace$$
Let $\ConfSet\in S$. We have: $\ConfSet = 
\underbrace{\ConfSet_{h_1}\cup\dots\cup\ConfSet_{h_k}}_{h_i\in\subsetatc{\finiteH}}\cup
\underbrace{\ConfSet_{h_1'}\cup\dots\cup\ConfSet_{h_n'}}_{h_i'\in\finiteH\setminus\subsetatc{\finiteH}}$.
Let us show that 
$\ConfSet_{h_1}\cup\dots\cup\ConfSet_{h_k}\cup
\ConfSet_{h_1'}\cup\dots\cup\ConfSet_{h_{n-1}'}\in S$.
By induction, this will prove that 
$\ConfSet_{h_1}\cup\dots\cup\ConfSet_{h_k}\in S$.
We have $h_n'\in\finiteH\setminus\subsetatc{\finiteH}$,
so there exists $h_i\in\subsetatc{\finiteH}$ such that $h_i\le h_n'$.
As $\ConfSet_{h_i}$ is a minimal $h_i$-predecessor of an element $\ConfSet'$ in
$\subsetatc{\LSafe{ua}}$, it follows 
from (\ref{eqn:min-hedges-and-pred}) that
$\ConfSet_{h_i}$ is also a minimal $h_n'$-predecessor of $\ConfSet'$.
So $\ConfSet_{h_1}\cup\dots\cup\ConfSet_{h_k}\cup
\ConfSet_{h_1'}\cup\dots\cup\ConfSet_{h_{n-1}'}\cup\ConfSet_{h_i}\in S$.
\end{proof}

We have also the next proposition.

\begin{proposition}
\label{prop:safe-a-finiteH}
$$
\subsetatc{\Safe{a}} = 
\subsetatc{ \left\lbrace \ConfSet \mid \forall h \in\subsetatc{\finiteH},
\exists q_f \in Q_f, \exists (q, \stack) \in \ConfSet :
(q, \stack)  \xrightarrow{h\ol{a}}  (q_f, \epsilon) \right\rbrace}.
$$
\end{proposition}

\subsubsection{An appropriate union operator}

Proposition~\ref{prop:safe-lsafe-atc-Hatc} expresses that every 
set of configurations $\ConfSet$ in $\subsetatc{\Safe{ua}}$ is the union
of $\ConfSet_h$ with $h\in\subsetatc{\finiteH}$.
We introduce a new operator to improve the readability and find new
properties.

\begin{definition}
Let $S$ be a finite set, and $A,B \in  \in 2^{2^S \setminus \{\emptyset\}}$.
The set $A\roundcup B \in 2^{2^S}$ is defined by:
$$
A \roundcup B = \{ a\cup b \hmid a\in A\text{ and } b\in B \}
$$
\end{definition}
Operator $\roundcup$ builds sets obtained by taking one set of each of its
operands, and performing their union. 
It is obviously associative and commutative.
Notice that the elements of $A, B$ are supposed to be non-empty sets. 
This will always be the case in the following algorithms using this operator.
Proposition~\ref{prop:safe-lsafe-atc-Hatc} can now be rewritten as follows.

\begin{proposition} \label{prop:Step2Improvedbis}
\label{prop:safe-lsafe-atc-Hatc-roundcup}
$$
\subsetatc{\Safe{ua}} =
\subsetatc{
\bigsqcup_{h\in\subsetatc{\finiteH}}
\left\lbrace
\ConfSet_h \hmid 
\ConfSet_h \text{ is a minimal $h$-predecessor of } 
\ConfSet'\in\subsetatc{\LSafe{ua}}
\right\rbrace
}
$$
\end{proposition}

When combined with operator $\subsetatc{.}$,
clauses of the $\roundcup$ operator can be splitted,
so that $\roundcup$ is to be computed on smaller sets.

\begin{lemma}
\label{lem:decompose-roundcup} 
$\subsetatc{A \roundcup B} = 
\subsetatc{(A \cap B)  \cup  
(A\setminus B\ \roundcup \ B\setminus A)}$
\end{lemma}

\begin{proof}
$(\supseteq)$
Let $\ConfSet \in \subsetatc{(A \cap B)  \cup  
(A\setminus B\ \roundcup \ B\setminus A)}$.
Then $\ConfSet\in A\roundcup B$.
For contradiction, let us assume that there exists $\ConfSet'\subsetneq\ConfSet$
such that $\ConfSet'\in\subsetatc{A\roundcup B}$.
If $\ConfSet'\in A\cap B$ then $\ConfSet'\in (A \cap B)  \cup  
(A\setminus B\  \roundcup\  B\setminus A)$,
which contradicts $\ConfSet$.
So $\ConfSet'\notin A\cap B$, and assume wlog that $\ConfSet'=a\cup b$
with $a\in A\setminus B$ and $b\in B$.
If $b\in A$ then $b\in A\cap B\subseteq A\roundcup B$ and 
$b\subsetneq\ConfSet'$, but this contradicts $\ConfSet'$.
If $b\notin A$ then $\ConfSet'\in A\setminus B\roundcup B\setminus A$,
so $\ConfSet'\in A\cap B \cup (A\setminus B\ \roundcup \ B\setminus A)$,
and $\ConfSet'\subsetneq\ConfSet$, which contradicts $\ConfSet$.

$(\subseteq)$
Let $\ConfSet\in\subsetatc{A\roundcup B}$.
Let us first show that 
$\ConfSet\in (A\cap B) \cup (A\setminus B\ \roundcup\ B\setminus A)$.
If $\ConfSet\in A\cap B$ this is direct. 
Otherwise $\ConfSet=a\cup b$ with $a\in A\setminus B$ and $b\in B$
(the other case is symmetric).
If $b\in A$ then $b\in A\cap B\subseteq A\roundcup B$ and $b\subsetneq\ConfSet$,
which contradicts the definition of $\ConfSet$.
So $b\in B\setminus A$, and $\ConfSet\in A\setminus B \ \roundcup \ B\setminus A$.
Now, assume for contradiction that there exists $\ConfSet'\subsetneq\ConfSet$
such that 
$\ConfSet'\in\subsetatc{(A\cap B) \cup (A\setminus B \ \roundcup \ B\setminus A)}$.
Then, according to $(\supseteq)$, $\ConfSet'\in\subsetatc{A\roundcup B}$,
which contradicts the definition of $\ConfSet$.
\end{proof}

\begin{corollary}
If $A\subseteq B$, then $\subsetatc{A \roundcup B} = \subsetatc{A}$.
\end{corollary}

The $\roundcup$ operator also simplifies the definition of
$\lfloor \Safe{a} \rfloor$. From this new definition, an algorithm follows.

\begin{proposition}
$\subsetatc{\Safe{a}} = 
\subsetatc{ \bigsqcup_{h\in\subsetatc{\finiteH}} A_h }$
with 
$$A_h = \left \{ \{ (q,\sigma) \} \mid
q \in Q, \sigma \in \Gamma : \exists q_f \in Q_f :  (q,\sigma) \xrightarrow{h{\ol a}} 
(q_f, \epsilon) \right \}.$$
\end{proposition}

\begin{proof}
\begin{enumerate}
\item Every element of $\bigsqcup_{h\in\subsetatc{\finiteH}} A_h$ belongs to $\Safe{a}$. Thus $\subsetatc{\bigsqcup_{h\in\subsetatc{\finiteH}} A_h}$ $\subseteq \Safe{a}$.
\item Let us show that for each $\ConfSet$ in $\Safe{a}$, there exists $\ConfSet' \in \bigsqcup_{h\in\subsetatc{\finiteH}} A_h$ such
that $\ConfSet' \subseteq \ConfSet$.
Let $\ConfSet \in \Safe{a}$. By definition, for all $h \in \subsetatc{\finiteH}$
there exists $(q_h, \stack_h) \in \ConfSet$ and $q_f \in Q_f$ such that 
$(q_h, \stack_h)  \xrightarrow{h\ol{a}}  (q_f, \epsilon)$.
Let $\ConfSet' = \{ (q_h, \stack_h) \mid h  \in \subsetatc{\finiteH}  \}$. Then
$\ConfSet' \subseteq \ConfSet$ and $\ConfSet' \in \bigsqcup_{h\in\subsetatc{\finiteH}} A_h$
because $\{ (q_h,\stack_h) \} \in A_h, \forall h$.
\item Assume that there exists $\ConfSet_* \in \subsetatc{\bigsqcup_{h\in\subsetatc{\finiteH}} A_h} \setminus \subsetatc{\Safe{a}}$. By 1., there exists $\ConfSet$ in $\subsetatc{\Safe{a}}$ such that $\ConfSet \subsetneq \ConfSet_*$; and by 2., there exists $\ConfSet' \in \bigsqcup_{h\in\subsetatc{\finiteH}} A_h$ such that $\ConfSet' \subseteq \ConfSet \subsetneq \ConfSet_*$ in contradiction with the definition of $\ConfSet_*$. Therefore $\subsetatc{\bigsqcup_{h\in\subsetatc{\finiteH}} A_h} \subseteq \subsetatc{\Safe{a}}$.
\item Let $\ConfSet \in \subsetatc{\Safe{a}}$. By 2., there exists $\ConfSet' \in \subsetatc{\bigsqcup_{h\in\subsetatc{\finiteH}} A_h}$ such that $\ConfSet' \subseteq \ConfSet$. By 3., it follows that $\ConfSet = \ConfSet'$ and thus $\subsetatc{\Safe{a}} \subseteq \subsetatc{\bigsqcup_{h\in\subsetatc{\finiteH}} A_h}$.

\end{enumerate}
\end{proof}

\subsubsection{Using SAT solvers to find minimal predecessors}

The computation of minimal predecessors is the key operation
for Step~1 and Step~2 which respectively compute $\subsetatc{\LSafe{ua}}$ from $\subsetatc{\Safe{u}}$ and $\subsetatc{\Safe{ua}}$ from $\subsetatc{\LSafe{ua}}$ using the following formulas (see Corollary~\ref{cor:Step1Antichains} and Proposition~\ref{prop:Step2Improvedbis}) :
$$
\subsetatc{\LSafe{ua}} =
\subsetatc{\left\lbrace\ConfSet \hmid \ConfSet 
\text{ is a minimal $\ol a$-predecessor of } \ConfSet'\in\subsetatc{\Safe{u}}
\right\rbrace}
$$
$$
\subsetatc{\Safe{ua}} =
\subsetatc{
\bigsqcup_{h\in\subsetatc{\finiteH}}
\left\lbrace
\ConfSet_h \hmid 
\ConfSet_h \text{ is a minimal $h$-predecessor of } 
\ConfSet'\in\subsetatc{\LSafe{ua}}
\right\rbrace
}
$$

We propose a method to compute minimal predecessors by performing
multiple calls to a SAT solver. A SAT solver is an algorithm used to efficiently test the satisfiability of a boolean formula $\varphi$, that is to check whether there exists a valuation $\charactval$ of the boolean variables of $\varphi$ that makes $\varphi$ true. In this case we say that $v$ is a \emph{model} of $\varphi$, denoted by $\charactval \models \varphi$.

Most of the SAT solvers require that the boolean formula given as input is
a conjunction of clauses (where a clause is a disjunction of literals, and
a literal is a variable or its negation).
Such formulas are said to be in conjunctive normal form (CNF).
In the following all input formulas will be in CNF.

We first detail a method to compute all minimal $\ol a$-predecessor of $\ConfSet'$. It is also valid to compute all minimal $h$-predecessors of $\ConfSet'$. 

\paragraph{Minimal predecessors.}
We recall that $\ConfSet$ is a $\ol a$-predecessor of
$\ConfSet'$ if for all $(q',\stack')\in\ConfSet'$, 
there exists $(q,\stack)\in\ConfSet$ such that
$(q,\stack) \xrightarrow{\ol a} (q',\stack')$.
Let us write $\satform{\ol a}{\ConfSet'}$ for the following boolean formula:
$$
\satform{\ol a}{\ConfSet'} = 
\bigwedge_{c'\in\ConfSet'} \bigvee_{c\xrightarrow{\ol a} c'} x_c,
$$
and let $\charact{\ConfSet}$ be the valuation such that $\charact{\ConfSet}(x_c) = 1$ iff $c\in\ConfSet$.
Then we immediately obtain that:
$$
\charact{\ConfSet}\models \satform{\ol a}{\ConfSet'}
\text{\qquad iff \qquad
$\ConfSet$ is an $\ol a$-predecessor of $\ConfSet'$}
$$

We define an ordering over valuations as follows, in a way to have a notion of minimal models equivalent to minimal predecessors.
Let $\varphi$ be a CNF boolean formula over the set $V$ of boolean variables, let $\charactval$ and $\charactval'$ be two valuations over $V$. We define $\charactval' \le \charactval$ iff for all variables $x\in V$, 
$\charactval'(x) = 1 \implies \charactval(x) = 1$. 
We denote  $\charactval' < \charactval$ if $\charactval' \le \charactval$ and $\charactval' \not = \charactval$.
We say that a model $\charactval$ of $\varphi$ is \emph{minimal} if for all model $\charactval'$ of $\varphi$, we have $\charactval' \le \charactval \implies \charactval' = \charactval$. We get the next characterization which also holds for $h$-predecessors.

\begin{lemma}
\label{lem:predecessors-using-sat} 
$\ConfSet$ is a minimal $\ol a$-predecessor of $\ConfSet'$ iff
$\charact{\ConfSet}$ is a minimal model of $\satform{\ol a}{\ConfSet'}$.
\end{lemma}

We can now explain how to compute all the minimal $\ol a$-predecessors of $\ConfSet'$,
or equivalently all the minimal models of formula $\satform{\ol a}{\ConfSet'}$.

Let $\varphi$ be a CNF boolean formula over $V$. First, we explain, knowing a model $\charactval$ of $\varphi$, how to compute a model $\charactval'$ of $\varphi$ such that $\charactval' < \charactval$ (if it exists). Consider the next formula $\varphi'$:
$$
\varphi' = \varphi  \wedge
( \bigwedge_{x \in V_0} {\neg x}  )
\wedge ( \bigvee_{x \in V_1} {\neg x} )
$$
where $V_0$  (respectively $V_1$) is the set of all variables $x \in V$
such that $\charactval(x) = 0$ (resp. $\charactval(x) = 1$). If $\varphi'$ has a model $\charactval'$, it follows from the definition of $\varphi'$ that $\charactval'$ is a model of $\varphi$ such that $\charactval' < 
\charactval$. Otherwise, $\charactval$ is a minimal model of $\varphi$.
So from a model of $\varphi$ we can compute a minimal model of $\varphi$
by repeating the above procedure.

Second, let us explain how to compute all the minimal models of $\varphi$. Suppose that we already know some minimal model $\charactval$ of $\varphi$, and let $V_1$ be the set of variables $x \in V$ such that $\charactval(x) =1$. Consider the formula
$$ \varphi' = \varphi \wedge ( \bigvee_{x \in V_1} \neg x ).$$
Then a model $\charactval'$ of $\varphi'$, if it exists, is a model of $\varphi$ such that neither $\charactval' < \charactval$ (since $v$ is minimal) nor $\charactval < \charactval'$ (by definition of $\varphi'$). 
With the previous procedure, we thus get a minimal model of $\varphi$ that is distinct from $\charactval$. In this way we can compute all minimal models of $\varphi$.

This approach has been detailed for minimal $\ol a$-predecessors. It also works for minimal $h$-predecessors.

\paragraph{Step 1 with SAT solvers.}
The computation of the set $\subsetatc{\LSafe{ua}}$ from $\subsetatc{\Safe{u}}$ can also be done using SAT solvers. Indeed, suppose that given $\ConfSet'_1 \in\subsetatc{\Safe{u}}$, we have computed all the minimal $\ol a$-predecessors of $\ConfSet'_1$ as explained before. Let $\ConfSet'_2$ be another elements of $\subsetatc{\Safe{u}}$. As done previously, we can express by boolean formulas, that we want to compute minimal $\ol a$-predecessor of $\ConfSet'_2$ that are either strictly included in some minimal $\ol a$-predecessor of $\ConfSet'_1$, or incomparable with all minimal $\ol a$-predecessors of $\ConfSet'_1$. 

\paragraph{Step 2 with SAT solvers.}
The computation of the set $\subsetatc{\Safe{ua}}$ from $\subsetatc{\LSafe{ua}}$ can be done as in Proposition~\ref{prop:Step2Improvedbis} by using operator $\roundcup$ and exploiting its properties. 

Under the hypothesis that $\epsilon \in \subsetatc{\finiteH}$, an alternative is possible with Proposition~\ref{prop:Step2Improved} stating that $\subsetatc{\Safe{ua}}$ is equal to
$$
\subsetatc{\left\lbrace\ConfSet \hmid \ConfSet=\bigcup_{h\in\subsetatc{\finiteH}} 
\ConfSet_h
\text{ with $\ConfSet_h$ a minimal $h$-predecessor of } 
\ConfSet'\in\subsetatc{\LSafe{ua}}\right\rbrace}
$$
It is based on the following observations. Fix some $\ConfSet$ and $\ConfSet'$ in the previous equality. First, if $h = \epsilon$, then $\ConfSet'$ is the only minimal $h$-predecessor of $\ConfSet'$ and thus $\ConfSet' \subseteq \ConfSet$. Second we know by the proof of Proposition~\ref{prop:finiteObjects} that $\ConfSet \subseteq Q \times \Gamma^{|u'|}$ where $u' = \open{u}$. Therefore, instead of computing $\ConfSet$ as a union $\bigcup_{h\in\subsetatc{\finiteH}} \ConfSet_h$, we can compute it starting from $\ConfSet'$ and adding elements of $Q \times \Gamma^{|u'|}$ one by one, until we get an element $\ConfSet$ of $\Safe{ua}$. By the way it is constructed, $\ConfSet \in \subsetatc{\Safe{ua}}$. We can check that such an element belongs to $\Safe{ua}$ with Proposition~\ref{prop:Post} by testing for all $h \in \subsetatc{\finiteH}$, whether there exists $\ConfSet'' \in \subsetatc{\LSafe{ua}}$ such that $\Post{\lin h}{\ConfSet} \supseteq \ConfSet''$. To get the whole set $\subsetatc{\Safe{ua}}$, we need to consider all the possibilities to enlarge $\ConfSet'$ with elements of $Q \times \Gamma^{|u'|}$. This task can be done efficiently with the help of SAT solvers (with ideas similar to the ones developed above).

\bibliographystyle{amsalpha}
\bibliography{bibliographie}

\end{document}